\documentclass[acmsmall,nonacm]{acmart}  

\usepackage{amsfonts}
\usepackage{microtype}
\usepackage{amsmath}
\usepackage{amsthm}
\usepackage{xspace}
\usepackage{xcolor}
\usepackage{hyperref}
\usepackage{graphicx}
\usepackage{url}
\usepackage{paralist}
\usepackage{verbatim}
\usepackage[final]{pdfpages}
\usepackage{enumitem}
\usepackage[normalem]{ulem}

\usepackage{thmtools}
\usepackage{thm-restate}

\newcommand{\TA}{T_{\mathcal{A}}}
\newcommand{\BA}{B_{\mathcal{A}}}
\newcommand{\eps}{\varepsilon}
\newcommand{\keps}{\lceil 1/\eps \rceil}

\newcommand{\congest}{\ensuremath{\mathsf{CONGEST}}\xspace}
\newcommand{\bcongest}{\ensuremath{\mathsf{BCONGEST}}\xspace}

\newcommand{\dil}{\ensuremath{\mathsf{dilation}}\xspace}
\newcommand{\congestion}{\ensuremath{\mathsf{congestion}}\xspace}

\def\polylog{\operatorname{polylog}}

\newtheorem{theorem}{Theorem}[section]
\newtheorem{lemma}[theorem]{Lemma}

\newtheorem{definition}[theorem]{Definition}

\newtheorem{corollary}[theorem]{Corollary}

\usepackage[utf8]{inputenc}

\usepackage{algorithm}
\usepackage{algorithmicx}
\usepackage{algpseudocode}
\algnewcommand{\IIf}[1]{\State\algorithmicif\ #1\ \algorithmicthen}
\algnewcommand{\EndIIf}{\unskip\ \algorithmicend\ \algorithmicif}
\algtext*{EndWhile} 
\algtext*{EndIf} 
\algtext*{EndFor} 
\algdef{SE}[SUBALG]{Indent}{EndIndent}{}{\algorithmicend\ }%
\algtext*{Indent}
\algtext*{EndIndent}

\usepackage{ifthen}

\newboolean{short}
\setboolean{short}{true}

\newcommand{\onlyShort}[1]{\ifthenelse{\boolean{short}}{#1}{}}
\newcommand{\onlyLong}[1]{\ifthenelse{\boolean{short}}{}{#1}}
\newcommand{\shortLong}[2]{\ifthenelse{\boolean{short}}{#2}{#1}}
\newcommand{\longShort}[2]{\ifthenelse{\boolean{short}}{#2}{#1}}

\begin{document}

\title{Message Optimality and Message-Time Trade-offs for APSP and Beyond}

\author{Fabien Dufoulon}
\orcid{0000-0003-2977-4109}
\affiliation{%
\institution{Lancaster University}
\city{Lancaster}
\country{UK}}
\email{f.dufoulon@lancaster.ac.uk}

\author{Shreyas Pai}
\orcid{0000-0003-2409-7807}
\affiliation{%
\institution{IIT Madras}
\city{Chennai}
\country{India}}
\email{shreyas@cse.iitm.ac.in}

\author{Gopal Pandurangan}
\orcid{0000-0001-5833-6592}
\affiliation{%
\institution{University of Houston}
\city{Houston}
\country{USA}}
\email{gopal@cs.uh.edu}

\author{Sriram Pemmaraju}
\orcid{0000-0002-0834-3476}
\affiliation{%
\institution{University of Iowa}
\city{Iowa City}
\country{USA}}
\email{sriram-pemmaraju@uiowa.edu}

\author{Peter Robinson}
\orcid{0000-0002-7442-7002}
\affiliation{%
\institution{Augusta University}
\city{Augusta}
\country{USA}}
\email{perobinson@augusta.edu}

\date{}

\thispagestyle{empty}

\begin{abstract}
  Round complexity is an extensively studied metric of distributed algorithms. In contrast, our knowledge of the \emph{message complexity} of distributed computing problems and its relationship (if any) with round complexity is still quite limited.
  To illustrate, for many fundamental distributed graph optimization problems such as (exact) diameter computation, All-Pairs Shortest Paths (APSP), Maximum Matching etc., while (near) round-optimal algorithms are known,  message-optimal algorithms are  hitherto unknown. More importantly, the existing round-optimal algorithms are not message-optimal. This raises two important questions: (1) Can we design message-optimal algorithms for these problems?  (2) Can we give message-time tradeoffs for these problems in case the message-optimal algorithms are not round-optimal?
  
  In this work, we focus on a fundamental graph optimization problem, \emph{All Pairs Shortest Path (APSP)}, whose message complexity is still unresolved. We present two main results in the \congest{} model:
  \begin{itemize}
  \item We give a message-optimal (up to logarithmic factors) algorithm that solves weighted APSP, using $\tilde{O}(n^2)$ messages. This algorithm takes $\tilde{O}(n^2)$ rounds.
  \item For any $0 \le \eps \le 1$, we show how to solve unweighted APSP in $\tilde{O}(n^{2-\eps})$ rounds and $\tilde{O}(n^{2+\eps})$ messages. At one end of this smooth trade-off, we obtain a (nearly) message-optimal algorithm using $\tilde{O}(n^2)$ messages (for $\eps = 0$), whereas at the other end we get a (nearly) round-optimal algorithm using $\tilde{O}(n)$ rounds (for $\eps = 1$). This is the first such message-time trade-off result known.
  \end{itemize}

Our first result is based on a simple, uniform simulation algorithm, that utilizes the Miller-Peng-Xu (MPX) graph decomposition [SPAA 2013], for designing message-efficient distributed algorithms. Our approach simulates any distributed algorithm with \emph{broadcast complexity} $B$ in (essentially) $\tilde{O}(B)$ message complexity. Note that the broadcast complexity of an algorithm is the number of broadcasts by all nodes over the entire execution, which can be significantly smaller than the algorithm's message complexity. 
We then use our simulation technique to derive the \emph{first-known message-optimal} distributed algorithms not only for \emph{weighted APSP}, but also for \emph{Maximum Matching in bipartite graphs}, and \emph{neighborhood cover construction}. 

Our second result, which is our main technical contribution, is based on a new technique that allows us to simulate many Breadth First Search (BFS) algorithms in parallel in a way that allows us to control the maximum congestion of an edge and the dilation (i.e., maximum running time) of the algorithms.  Our technique uses a modification  the well-known hierarchical cluster decomposition of Baswana-Sen [RSA 2007] in a novel way. The main novelty is running the simulation on an ensemble of \emph{independently} constructed cluster hierarchies. We believe that our technique can be useful in obtaining  message-time tradeoffs for other problems.

\end{abstract}

\maketitle

\section{Introduction}
\label{sec:intro}

\textit{Message} and \textit{time} complexities are two fundamental performance measures of distributed algorithms. Both complexity measures crucially influence the performance of a distributed algorithm. Time complexity measures the number of distributed ``rounds'' taken by the algorithm, and it directly determines the running time of the algorithm. Traditionally, keeping the time (round) complexity as small as possible has been an important goal. Message complexity, on the other hand, measures the {\it total number of messages} sent and received by all the processors during the algorithm.
In many applications, message complexity is the dominant cost, playing a major role in determining the running time as well as additional resources (e.g., energy, bandwidth, memory, etc.) expended by the algorithm. Hence, designing distributed algorithms with  low message complexity, sometimes even at the cost of an increased round complexity could be desirable.

Since both measures are essential, ideally, one would like to design distributed algorithms with \textit{simultaneously} optimal message and time complexities --- so-called \textit{singularly optimal} algorithms. However, if this is too difficult or even impossible,   then one would like to design distributed algorithms that are at least \textit{separately} message or time optimal. One would like to go even further and design  \textit{a family of algorithms} that smoothly trades off between the two measures. However till now, while  singularly-optimal algorithms are known for fundamental problems such as leader election \cite{kuttenjacm15} and minimum spanning trees \cite{PanduranganRS2017,elkin17_simpl_deter_distr_mst_algor}, there
are other fundamental problems such as diameter computation or (even, unweighted) All-Pairs Shortest Paths (APSP) where not much is known with respect to message complexity vis-a-vis  the round complexity.

This paper focuses on the message complexity of APSP, though some of our results apply to other problems such as Maximum Matching (MaxM) and computing Neighborhood Covers. 
APSP has been studied extensively for many decades in both sequential and distributed computing. Significant progress has been made towards designing distributed algorithms with near-optimal round complexity for APSP in the classical \congest{} model of distributed computing (see Section \ref{sec:model} for model definitions).
For example, after a series of improvements \cite{HuangNanongkaiSarunurakFOCS2017,ElkinSTOC2017}, Bernstein and Nanongkai\cite{bernstein2021distributed} presented an $\tilde{O}(n)$-round\footnote{$\tilde{\Omega}$ and $\tilde{O}$ hide a $1/\polylog{n}$ and $\polylog{n}$ factor respectively.} randomized (Las Vegas) algorithm for weighted APSP.
This is optimal within logarithmic factors because $\Omega(n)$ is a round complexity lower bound even for unweighted APSP \cite{Censor-HillelKP17}.
However, the optimum attainable message complexity is  unresolved for APSP. In particular, the Bernstein-Nanongkai \cite{bernstein2021distributed} round-optimal distributed algorithm mentioned above, has $\tilde{\Theta}(mn)$ message complexity (throughout, $m$ denotes the number of edges in the graph) which can be as large as  $\tilde{\Theta}(n^3)$ 
for dense graphs. Our work is motivated by two key questions: 
\begin{enumerate}
    \item Is $\tilde{\Theta}(n^3)$ the optimal (or near-optimal) message complexity for APSP or can we design distributed algorithms with significantly better (say, $\tilde{O}(n^2)$) message complexity? We note that using techniques from \cite{DufoulonITCS2024} applied to the lower bound graph in \cite{AbboudCHK16}, one can show an $\tilde{\Omega}(n^2)$ message lower bound for $\mathrm{poly}(n)$ round algorithms even for sparse graphs. 
    \item Is it possible to trade off rounds for messages for APSP? Specifically, can we design APSP algorithms that are more message-frugal (relative to the $\tilde{\Theta}(n^3)$ message bound of \cite{bernstein2021distributed}), but possibly use more rounds?
    \end{enumerate}

The thrust of the paper is to prove the two following theorems that answer the above questions.
Our first result shows that indeed, it is possible to solve weighted APSP in the \congest{} model using $\tilde O(n^2)$ messages, showing that message optimality (within logarithmic factors) is achievable for weighted APSP.
This result is a specific instance of a more general simulation result we prove (see Theorem \ref{thm:messageEfficientSimulation}) that shows that broadcast-based \congest{} algorithms that perform a total of $B$ broadcast operations (across all nodes) can be simulated with a message complexity of $\tilde{O}(B)$. 
Our result is significant because there are many examples of distributed algorithms that are broadcast-based (e.g., Breadth First Search, Luby's algorithm for Maximal Independent Sets \cite{LubySICOMP1986}) and the message complexity of these algorithms is typically quite a bit larger than number of broadcasts it performs. For example, the natural Breadth First Search (BFS) algorithm performs at most $n$ broadcasts, but has $\Theta(n^2)$ message complexity in the worst case (for dense graphs).
We also apply our simulation result to other problems, such as \textsc{MaxM} and Neighborhood Covers.
\begin{restatable}{theorem}{MessageOptimalAPSP}
\label{thm:apsp}
There is a $\congest$ algorithm that, with high probability, computes exact weighted APSP in $\tilde O(n^2)$ rounds and with $\tilde O(n^2)$ message complexity, even on directed graphs and even if the edge weights are negative.
\end{restatable}

Our second result applies to unweighted APSP and shows that for this problem, there is a natural, smooth trade-off, parameterized by $\eps \in [0, 1]$, between messages and rounds. 
At one end of this smooth trade-off, we obtain a (nearly) message-optimal algorithm using $\tilde{O}(n^2)$ messages (for $\eps = 0$).
For this point in the trade-off, the round complexity is $\tilde{O}(n^2)$ and this point in the trade-off can be viewed as a special case of Theorem \ref{thm:apsp}.
At the other end of the trade-off, we get a (nearly) round-optimal algorithm using $\tilde{O}(n)$ rounds (for $\eps = 1$). As far as we know, this is the first such message-time trade-off result in the \congest{} model, for any graph problem.

\begin{restatable}{theorem}{APSPTradeoff}
\label{thm:APSPTradeoff}
    For any $\eps \in [0,1]$, unweighted APSP can be solved (w.h.p.) in $\tilde{O}(n^{2-\eps})$ rounds and  $\tilde{O}(n^{2+\eps})$ messages (w.h.p.) in \congest{}.
\end{restatable}

\subsection{Model}
\label{sec:model}

\subsubsection{\congest{} model}
\label{sec:congest}
The \congest{} model \cite{peleg00} is a standard message-passing model in distributed computing where the input graph $G$ defines the distributed network with the nodes denoting machines and the edges denoting communication links between two machines. Each node has a unique \texttt{ID} drawn from a space whose size is polynomial in $n$. 
We assume  the \textit{synchronous} model, where both computation and communication proceed in lockstep, i.e., in discrete time steps called \textit{rounds}.
 In the $\congest$ model, we allow
only small message sizes (typically logarithmic in $n$, the number of nodes) to be sent per edge per round. 
Since each \texttt{ID} can be represented with $O(\log n)$ bits, each message in the $\congest$ model is large enough to contain a constant number of \texttt{ID}s. 
The $\congest$ model captures bandwidth limitations inherent in real-world networks.
In each round of the synchronous $\congest$ model, each node (i) receives all messages sent to it in the previous round, (ii) performs arbitrary local computation based on information it has\footnote{As is standard, the cost of local computation is ignored and, in each round, each node can perform an arbitrary, even exponential time, computation, using the information it possesses. 
This assumption is justified by the fact that communication cost dominates local computation cost substantially in many settings.}, and (iii) sends a \(O(\log n)\)-bit message to each of its neighbors in the graph. Note that a node may send \textit{different messages}  to \textit{different neighbors} in a round.  
For an algorithm $\mathcal{A}$ in the synchronous $\congest$ model, its \textit{round complexity} is the number of rounds it takes to finish and produce output and its \textit{message complexity} is the total number of messages sent by all nodes over the course of the algorithm. For brevity we drop ``synchronous'' and just call this the $\congest$ model.

\subsubsection{\bcongest{} model}
\label{sec:bcongest}
We define the \bcongest{} model, which is a variant of the \congest{} model. It only differs in that in each round, a node must send the {\it same} message to all of its neighbors. Just as in the \congest{} model, we define the \textit{message complexity} of a \bcongest{} algorithm as the number of messages sent by all nodes over the entire execution. However, we also introduce a related complexity measure, which we call \textit{broadcast complexity}: it is the number of broadcasts by all nodes over the entire execution. In this paper, we show that low broadcast complexity plays an important role in designing message efficient \congest{} algorithms. In particular, the broadcast complexity bounds the cost of simulating a \bcongest{} algorithm message efficiently in the \congest{} model. 

\subsection{Our Technical Contributions}
\label{sec:contributions}

Our contributions are two-fold. We present the first message-optimal algorithm for weighted APSP (Theorem \ref{thm:apsp}) and we present the first time-message trade-off for unweighted APSP (Theorem \ref{thm:APSPTradeoff}). We now describe our approach, at a high level, for obtaining these results.

\paragraph{Message-Optimal Algorithms.} 
In a nutshell, our approach for message-optimal algorithms uses the following general idea. First, suppose that we are able to partition the graph into clusters such that (i) each cluster has low diameter and (ii) each vertex has neighbors in only a small number of clusters. Note that requirement (ii) permits vertices to have high degree.
Now consider an algorithm $\mathcal{A}$ in the \bcongest{} model with low broadcast complexity. Our key observation is this: \textit{whenever a vertex $v$ performs a broadcast, instead of sending a message to all neighbors of $v$, we can send a message to one representative neighbor for each of the  neighboring clusters}. The fact that every vertex has only a small number of neighboring clusters implies that the number of messages sent is small. The fact that the diameter of each cluster is small implies that it is possible for 
the center of each cluster to efficiently (in terms of messages) serve as a proxy for all nodes in the cluster and thus it is enough for $v$'s message to reach a single node in each neighboring cluster.
A key ingredient of our approach is the notion of a \textit{Low Diameter and Communication (LDC) graph decomposition} (cf. Section \ref{subsec:LDC}).
A LDC decomposition structure allows us to pre-process the graph 
to obtain the clustering with properties described above.
This decomposition is a simple variant of the low-diameter graph decomposition algorithm of Miller, Peng, and Xu \cite{MPX13}.
Using this approach, we show how to obtain an equivalent \congest{} algorithm whose message complexity is within polylogarithmic factors of the simulated algorithm's broadcast complexity at the cost of a higher time complexity (roughly, by a linear in $n$ factor) --- cf. Theorem \ref{thm:messageEfficientSimulation}.

\paragraph{Message-Time Tradeoffs.} While the above approach yields a message-optimal algorithm for weighted APSP, this algorithm has a high round complexity --- $\tilde{O}(n^2)$ --- for unweighted APSP.
This is in fact, $O(n)$ factor higher than the round optimal algorithm known for unweighted APSP as discussed above. On the other hand, as mentioned earlier, known round-optimal algorithms are not message optimal; they have a $\Theta(n^3)$ message complexity. The LDC decomposition framework, discussed above, that led to message optimality does not help obtain a trade-off between these two extremes. We need a new framework and our main technical contribution in this paper is the use of several new and disparate ideas that lead to such a framework. Below, we overview these ideas and how they lead to our family of algorithms. 

We start with a natural idea: solve unweighted APSP by simulating a collection of $n$ Breadth First Search (BFS) algorithms $\mathcal{A}_1, \mathcal{A}_2, \ldots, \mathcal{A}_n$, initiated at each of the $n$ different nodes. 
\begin{itemize}
\item \textbf{Baswana-Sen cluster hierarchy.} In order to control the message complexity of the collection of BFS algorithms, we simulate them over a \textit{Baswana-Sen cluster hierarchy}.
This hierarchy of clusters was defined by Baswana and Sen \cite{BaswanaSenRSA2007} as part of their classical randomized spanner algorithm. 
Baswana and Sen present an algorithm that, for any input graph $G = (V, E)$ and any integer $\kappa \ge 1$, constructs $(2\kappa-1)$-spanner $H = (V, E_H)$ of $G$ with $O(n^{1+1/\kappa})$ edges.
While the final output of this algorithm is a spanning subgraph $H$ of $G$, a byproduct of this algorithm is a $(\kappa+1)$-level hierarchy of clusters (defined precisely in Section \ref{subsec:BaswanaSen}). It is this Baswana-Sen cluster hierarchy that we utilize for our simulation. 
Given an $\eps \in [0, 1]$, setting $\kappa = \lceil 1/\eps\ \rceil$ and computing a $(\kappa+1)$-level hierarchy of clusters, gives us, roughly speaking, the property that each node has neighbors in $\tilde{O}(n^\eps)$ clusters. 
We can then use the same idea as in our previous simulation, i.e., whenever a vertex $v$ performs a broadcast, instead of sending a message to all neighbors of $v$, we can send a message to one representative neighbor for each of the  neighboring clusters.

\item \textbf{``Congestion + Dilation'' framework.} 
A classical result of Leighton, Maggs, and Rao \cite{lmr} showed that if we are given
$\ell$ packets, where packet $p_j$ needs to be routed from a given source $s_j$ to a given target $t_j$ along a given path $\mathcal{P}_j$, then all of these packets can be scheduled, so as to take just
$\tilde{O}(\congestion + \dil)$ rounds, where \congestion is the maximum number of packets that are routed through an edge and \dil is the length of the longest path.
An extremely elegant idea of using \textit{random delays} to start the different packets, leads to their result.
Ghaffari \cite{ghaffarilmr} extended this result to a collection of arbitrary \congest{} algorithms (see Theorem \ref{thm:congestionPlusDilation}).
In order to bound the round complexity of our collection of BFSs using this ``Congestion + Dilation'' result, we need an efficient simulation over
a Baswana-Sen cluster hierarchy with good bounds on both the maximum congestion of an edge and the dilation (i.e., maximum running time) of the BFS algorithms.
In order to obtain an efficient simulation, we need to schedule our collection of BFS algorithms to have an additional property, namely, in any round, each node receives messages from at most $O(\log n)$ distinct BFS algorithms. We show (in Theorem \ref{thm:congestPlusDilationBFS}) that applying the random delay technique to a collection of BFS algorithms, not only gives the $\tilde{O}(\congestion + \dil)$ bound on the running time, but also this additional property.
This additional property allows us to view our collection of algorithms as an \textit{aggregation-based algorithm} and plays a crucial role in ensuring the simulation is both message and round efficient.
\item \textbf{Simulating aggregation-based algorithms.} A key reason for the efficiency of our simulation is the fact that we are simulating \textit{aggregation-based algorithms} (see Definition \ref{def:aggregationBasedAlgorithm}). To understand the definition of an aggregation-based algorithm, imagine that messages to a node $v$ at a particular round $r$, arrive in separate batches. In an aggregation-based algorithm, the next local state of $v$, can be computed by processing each batch separately, with all intermediate computation results being small in size. Algorithms that use functions such as $\min, \max, \text{sum}$, etc. fall in this category and for such functions it is clear that messages can be batched and intermediate results are small. Individual BFS algorithms also have the nice property of being aggregation-based.
However, a collection of BFS algorithms may not have this property because each batch of messages may contain information about many different BFSs and this information cannot be aggregated to have a small size. This is where the additional property of BFS scheduling mentioned in the previous item (``Congestion + Dilation'' framework) turns out to be useful. 

In the new simulation built upon the Baswana-Sen cluster hierarchy, a cluster can no longer maintain complete knowledge of all its cluster nodes' states efficiently, as a node may belong to multiple clusters across the levels of the cluster hierarchy. Therefore, instead, each cluster (center) routes messages on behalf of its (non-unique) cluster nodes. But this can lead to a significant increases in both runtime and message overheads, where the former comes about due to high congestion over certain edges of the cluster hierarchy. To remedy that, we restrict the simulation to aggregation-based algorithms, which now allows each cluster to compress all the messages it routes to a single destination into $\tilde{O}(1)$ bits, and logarithmically as many messages. This leads to better congestion bounds, and thus good runtime overhead, as well as better message complexity, through the following argument: informally, each node only needs to receive one message per cluster in its neighborhood, which is, roughly speaking, bounded to $O(n^{\eps})$ within a Baswana-Sen cluster hierarchy.  
\item \textbf{Smoothing congestion using an ensemble of Baswana-Sen cluster hierarchies.} Using just a single Baswana-Sen cluster hierarchy leads to excessive congestion on a few edges and very little congestion on many edges. 
A key idea to bypass this congestion bottleneck is to use not one, but a small number $\zeta = n^{\eps}$ of independently constructed Baswana-Sen cluster hierarchies. 
Then the set of $n$ BFS algorithms are partitioned into $\zeta$ equal-sized batches and each batch can be assigned a different Baswana-Sen cluster hierarchy. 
This simple and natural idea plays a crucial role in reducing the maximum congestion of an edge (see Lemma \ref{lem:congestionSmoothing}). A key reason why this smoothing property holds is that the Baswana-Sen cluster hierarchy has the helpful property that any edge in the input graph has a small probability of being chosen as a cluster edge (see Lemma \ref{lemma:clusterEdgesRare}). 
\item \textbf{Using ``landmark'' nodes.}
Even with the ideas described above, we do not completely achieve the trade-off we seek. Specifically, the dilation of our simulated BFS algorithms algorithms is too high.
One final idea we use is to terminate the $n$ BFSs at a depth of $O(n^{1-\eps})$. This allows the BFSs to discover short paths, i.e., shortest paths of length $O(n^{1-\eps})$. But this leaves us with all pairs of nodes that are further apart from each other. However, for these, we can take a more brute force approach based on the idea of sampling ``landmark'' nodes, that will generate BFS trees containing all the remaining shortest paths (i.e., between any two far away nodes). 
\end{itemize}

\subsection{Additional Related Work}
\label{sec:related}

The \bcongest{} model is used in this paper mainly as a vehicle for obtaining a message-efficient simulation for algorithms in the \congest{} model. However, there are a few papers that focus on designing algorithms and proving lower bounds in the \bcongest{} model.
For example, Chechik and Mukhtar \cite{ChechikMukhtarDISC2019} 
focus on the single-source reachability problem and the single-source shortest path problem
for graphs with diameter 1 and 2.
In \cite{ChechikMukhtarPODC2020}, the same authors present a new randomized algorithm for solving the weighted, directed variant of
the Single-Source Shortest Path problem in the \bcongest{} model using only $\tilde{O}(\sqrt{n} \cdot D^{1/4} + D)$ rounds, taking a big step towards closing the gap with respect to the lower bound.
There is additional work in the \bcongest{} model, but with the communication network restricted to being a clique and the input graph being an arbitrary subgraph of the communication network.
In this so called Broadcast Congested Clique model, there are both upper bounds \cite{KorhonenROPODIS.2017,HolzerPOPODIS2015} and lower bounds \cite{pai19_connec_lower_bound_broad_conges_clique,DruckerKOPODC2014}.

There has been a lot of research on APSP, but from a round complexity point of view.
For example, for the problem of computing the diameter, an $\Omega(n/\log n)$ lower bound is shown in \cite{FrischknechtHW12}, even for graphs with constant diameter. 
A matching $O(n/\log n)$ upper bound is proved for this problem in \cite{HuaSPAA2016}.
For the more general, APSP problem the authors of \cite{AbboudCKPTALG21} show
an $\Omega(n)$ lower bound for weighted APSP, thus separating its round complexity from unweighted APSP by at least a logarithmic factor.
In a breakthrough result, Bernstein and Nanongkai shows that this lower bound can be almost-exactly matched, by presenting a Las Vegas $\tilde{O}(n)$-round algorithm for weighted APSP~\cite{bernstein2021distributed}.
\subsection{Technical Preliminaries}

\subsubsection{The Congestion Plus Dilation Framework}
\label{section:congestionPlusDilation}
 
Consider a scenario in which we want to run $\ell$ independent distributed algorithms $\mathcal{A}_1, \mathcal{A}_2, \ldots, \mathcal{A}_\ell$ together in the \congest{} model. 
Following notation in \cite{ghaffarilmr}, we use \dil{} to refer to the maximum running time (number of rounds) of any of the algorithms $\mathcal{A}_j$, when executed in isolation.
For any edge $e$, let $c_j(e)$ be the number of rounds in which algorithm $\mathcal{A}_j$ sends a message over $e$ and let 
$\congestion(e) = \sum_{j=1}^{\ell} c_j(e)$.
Thus $\congestion(e)$ denotes the total number of messages sent over edge $e$ over all $\ell$ algorithms $\mathcal{A}_j$.
Finally, let $\congestion = \max_{e} \congestion(e)$. 

We will make repeated use of the following result due to Ghaffari \cite{ghaffarilmr} that shows that the $\ell$ algorithms can be scheduled in such a way that 
it takes essentially, only $\congestion + \dil$ rounds for all the algorithms to complete.
This result extends the classical result of Leighton, Maggs, and Rao \cite{lmr} that applied to the special case in which each algorithm $\mathcal{A}_j$ performed a routing task, routing a packet from a given source $s_j$ to a given target $t_j$ along a pre-specified path $\mathcal{P}_j$.
\begin{theorem}[Ghaffari \cite{ghaffarilmr}]
\label{thm:congestionPlusDilation}
It is possible to schedule algorithms $\mathcal{A}_1, \mathcal{A}_2, \ldots, \mathcal{A}_\ell$ together in the \congest{} model, using only private randomness such that, with high probability, all the algorithms complete in $O(\congestion + \dil \cdot \log n)$ rounds, after $O(\dil \cdot \log^2 n)$ rounds of pre-computation.
\end{theorem}

We also make use of a version of Theorem \ref{thm:congestionPlusDilation} that applies to scheduling a collection of \textit{partial BFS} algorithms (see Theorem \ref{thm:congestPlusDilationBFS} below). By a partial BFS algorithm, we just mean a BFS algorithm that, by design, might terminate early, even before it explores the entire graph. Our theorem assumes that we are dealing with the standard BFS algorithms, i.e., the algorithm in which each node broadcasts just once, on first receiving a ``BFS exploration'' message. Our theorem applies in the \bcongest{} model, i.e., the final algorithm, which consists of an efficient scheduling of the $\ell$ partial BFS algorithms, is also in the \bcongest{}. This scheduling also has a useful additional property that we prove, which is that in any round, at most $O(\log n)$ algorithms (out of $\ell$) are ``active'' in a neighborhood.

\begin{theorem}
\label{thm:congestPlusDilationBFS}
Consider a collection of $\ell \le n$ BFS algorithms $\mathcal{A}_1, \mathcal{A}_2, \ldots, \mathcal{A}_\ell$,
each initiated by a different node. 
Then it is possible to schedule all the BFS algorithms
together in the \bcongest{} model, using shared randomness such that, with high probability, (i) all the algorithms complete in $\tilde{O}(\ell + \dil)$ rounds and (ii) every node receives messages from at most $O(\log n)$ distinct BFS algorithms in any round.
\end{theorem}
\begin{proof}
(i) Execute each algorithm $\mathcal{A}_j$ as is, except that we start $\mathcal{A}_j$ after a \textit{random delay} $d_j$ chosen uniformly at random from the (integer) range $[1, \ell]$. Consider a node $v$, a round $r$, and a BFS algorithm $\mathcal{A}_j$.
The probability that $v$ broadcasts in round $r$ is $\le 1/\ell$,
implying that the expected number of algorithms for which $v$ broadcasts in round $r$ is at most 1. By a simple application of a Chernoff bound, with high probability, there are at most $O(\log n)$ algorithms (out of $\ell$) such that $v$ broadcasts in round $r$ in these algorithms.
By a union bound over all $n$ nodes and all $\ell + dilation$ rounds, we get that for every node $v$ and every round $r$, $v$ broadcasts in round $r$ in at most $O(\log n)$ algorithms. By expanding each round into a phase with $O(\log n)$ rounds, we can schedule all broadcasts in a conflict-free manner. This leads to a running time of $\tilde{O}(\ell + \dil)$ with high probability.

(ii) As in the proof of (i), consider a node $v$, a round $r$, and a BFS algorithm $\mathcal{A}_j$. Because of how the BFS algorithm works, any two nodes in the (closed) neighborhood of $v$ are within distance 2 of each of other, and broadcast within 2 rounds of each other. Thus, there is an interval $I_v$ of length 3 such that all nodes in the closed neighborhood of $v$ broadcast in a round from $I_v$ in algorithm $\mathcal{A}_j$. The event that some node in the closed neighborhood of $v$ broadcasts in round $r$ is equivalent to the event that $r \in I_v$ and this happens with probability at most $3/\ell$. By a similar argument as in (i), we see that with high probability, for every node $v$ and every round $r$, some node in the closed neighborhood of $v$ broadcasts in round $r$ in at most $O(\log n)$ algorithms. 
\end{proof}

\subsubsection{Distributed Primitives: Upcast and Downcast}

We present two classical distributed primitives, upcast and downcast over forests, that we will use in Section \ref{sec:simulation}. 
We start by describing the upcast primitive. We assume that each node $v$ has some input $in(v)$, and the upcast primitive must ensure that the root receives all nodes' input information. This is done as follows. All nodes send $in(v)$, cut up in $O(\log n)$ bit messages, to their parent. Additionally, each non-leaf node stores any message received from its children, and retransmit these messages to its own parent whenever possible. 

\begin{lemma}
\label{lem:upcast}
    Let $In \geq n$ be the bits of input information over all nodes. Then, upcast over a forest depth $d$ takes $O(In/\log n)$ rounds and $O(d \cdot In/\log n)$ messages.
\end{lemma}

\begin{proof}
    Each edge may transmit at most $O(In)$ bits of information (from a leaf node towards the root) throughout the execution. Since communication is done via $O(\log n)$ bit messages, at most $O(In/\log n)$ messages are sent through any given edge. A simple pipelining argument shows that upcast takes $O(In/\log n + d) = O(In / \log n)$ rounds.

    As for the message complexity, it suffices to show that at most $O(In / \log n)$ messages are generated initially, and any such message originating at some node $v$ induces at most $d$ messages throughout the upcast execution, one on each edge from $v$ to the root. The lemma statement follows.
\end{proof}

Next, we describe the downcast primitive. In it, the root $r$ starts with some set $M(r)$ of $O(\log n)$ bit messages, each with an associated destination node, and the downcast primitive must ensure that each message is received by its destination node. The root sends each message to the child that is the root of the subtree containing the message's destination node, whenever possible. Non-root nodes similarly transmit any received messages to their children if they are not the received message's destination node.

\begin{lemma}
\label{lem:downcast}
    Let $M$ be the set of messages over all roots. Then, downcast over a forest of depth $d$ takes $O(|M| + d)$ rounds and $O(d |M|)$ messages.
\end{lemma}

\begin{proof}
    Each root may downcast up to $|M|$ messages. This ensures that at most $|M|$ messages are sent via each edge, and a simple pipelining argument implies the runtime. As for the message complexity, each of the $|M|$ messages originating at some root induces at most $d$ messages, one per edge from the root to the originating message's destination node
\end{proof}
\section{Message-Efficient Simulation of \bcongest{} Algorithms}
\label{sec:simulation}

In this section, we show how we can message-efficiently simulate \bcongest{} algorithms whose message complexity is significantly higher than their broadcast complexity. More formally, we show how to obtain an equivalent \congest{} algorithm whose message complexity is within polylogarithmic factors of the simulated algorithm's broadcast complexity, at the cost of a higher time complexity (roughly, by a linear in $n$ factor) --- see Theorem \ref{thm:messageEfficientSimulation} below.

\begin{theorem}
\label{thm:messageEfficientSimulation}
    Let $\mathcal{A}$ be any \bcongest{} algorithm with $\TA$ round complexity, $\BA$ broadcast complexity, and let $In$ and $Out$ denote respectively the size of the inputs (including communication graph and IDs) and outputs over all nodes. Then, there exists a randomized (Monte Carlo) \congest{} algorithm $\mathcal{A'}$ simulating $\mathcal{A}$ with 
    round complexity $\tilde{O}(In + Out + \TA \, n)$ and message complexity $\tilde{O}(In + Out + \BA)$, with high probability.
\end{theorem}

Note that by definition $\BA \leq \TA \, n$ for any \bcongest{} algorithm $\mathcal{A}$, and thus we can get the following simpler statement as a corollary of Theorem \ref{thm:messageEfficientSimulation}.

\begin{corollary}
\label{cor:messageEfficientSimulation}
    Let $\mathcal{A}$ be any \bcongest{} algorithm with $\TA$ round complexity and let $In$ and $Out$ denote respectively the size of the inputs (including communication graph and IDs) and outputs over all nodes. Then, there exists a randomized (Monte Carlo) \congest{} algorithm $\mathcal{A'}$ simulating $\mathcal{A}$ with 
    round and message complexity $\tilde{O}(In + Out + \TA \, n)$, with high probability.
\end{corollary}

Next, we describe the section's organization. We first describe a decomposition structure (called LDC decomposition) that allow us to pre-process the graph in Subsection \ref{subsec:LDC}. Then, we show how to leverage this preprocessing to simulate \bcongest{} algorithms in a message efficient manner, in Subsection \ref{subsec:MPXsimulationDescription}. Next, we analyze the simulation in Subsection \ref{subsec:MPXSimulationAnalysis}. Finally, we give a summary of some applications of our simulation in Subsection \ref{subsec:simulationApplications1}, which includes achieving message-optimality for weighted APSP.

\subsection{Low Diameter and Communication (LDC) Graph Decompositions}
\label{subsec:LDC}

Low diameter graph decompositions play a key role in distributed computing, and in particular for designing faster algorithms for local distributed problems (e.g., maximal independent set and coloring). 
For the sake of improved message complexity, we introduce low diameter and communication (LDC) graph decompositions. Informally, such decompositions also bound the number of edges going out of any one node of a cluster into all neighboring clusters. 

\begin{definition}[Low Diameter and Communication (LDC) Graph Decomposition]
    Let $G = (V,E)$ be an unweighted graph. A $(r,d)$-low diameter and communication decomposition of $G$ is a partition of the vertex set $V$ into subsets $V_1,\ldots,V_k$, called clusters, combined with a sparse inter-cluster communication directed edge set $F \subseteq E$ such that:
    \begin{itemize}
        \item Each cluster $V_i$ has strong diameter at most $r$, i.e., $dist_{G[V_i]}(u,v) \leq r$ for any two nodes $u,v \in V_i$,
        \item Each node $v \in V_i$ has at most $d$ (outgoing) incident edges in $F$, such that for each cluster $V_{j}$ containing a neighbor of $v$, there exists at least one incident edge $e = (v,u) \in F$ with $u \in V_{j}$. 
    \end{itemize}
\end{definition}

\begin{figure}
    \centering
    \includegraphics[width=0.7\linewidth]{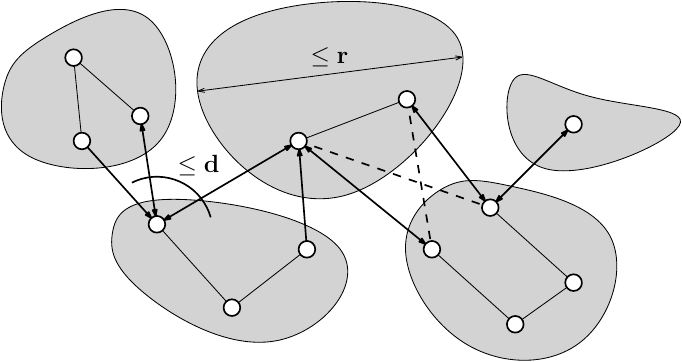}
    \caption{Example of an $(r,d)$-LDC decomposition with five clusters. The bold inter-cluster edges are the edges in $F$ whereas the dashed intercluster edges are in $E \setminus F$. At a high level, the simulation described in Section~\ref{sec:simulation} saves messages by not using the edges in $E \setminus F$.  Each cluster has strong diameter at most $r$, and each node in a cluster has at most $d$ outgoing edges in $F$.}
    \label{fig:LDC-example}
\end{figure}

Next, we show that a low diameter and low communication graph decomposition can be obtained with some simple additions to the low diameter graph decomposition algorithm of Miller, Peng, and Xu \cite{MPX13} which is popularly referred to as the MPX algorithm.

\begin{lemma}
\label{lem:LDCGraphDecomposition}
    There exists an $O(\log n)$ round algorithm that produces an $(O(\log n),O(\log n))$ low diameter and communication graph decomposition. Moreover, each cluster of that decomposition is spanned by a tree of depth $O(\log(n))$.
\end{lemma}

\begin{proof}
    The MPX algorithm of \cite{MPX13} takes $O(\log n)$ rounds and produces a low diameter graph decomposition of $G$, that is, it partitions the vertex set $V$ into subsets $V_1,\ldots,V_k$ such that each cluster $V_i$ has strong diameter at most $O(\log n)$ and is spanned by a tree of depth $O(\log n)$. Moreover, Corollary 3.9 of \cite{HW16} proves that with high probability, each node $v \in V$ neighbors $O(\log n)$ clusters, or in other words, the neighbors of $v$ belong to at most $O(\log n)$ clusters. Thus, by adding for each such neighboring cluster, one edge from $v$ to a neighbor in that cluster to the edge set $F$, we indeed obtain an $(O(\log n),O(\log n))$ low diameter and communication decomposition of $G$.
\end{proof}

\subsection{Description of the Simulation}
\label{subsec:MPXsimulationDescription}

Let $\mathcal{A}$ denote an arbitrary \bcongest{} algorithm with $\TA$ round complexity and $\BA$ broadcast complexity, and let $In$ and $Out$ denote respectively the size of the inputs (including communication graph and IDs) and outputs over all nodes.
Then, we give a \congest{} algorithm $\mathcal{A'}$ that simulates $\mathcal{A}$ (i.e., produces the same output). The algorithm  $\mathcal{A'}$ can be decomposed into two parts: a preprocessing part, and a simulation part. In the preprocessing part, we compute a (low diameter and communication) clustering of the communication graph which will be crucial in limiting the message complexity of the simulation part to roughly the broadcast complexity of the simulated algorithm $\mathcal{A}$.

\paragraph{Preprocessing.}
The preprocessing part of $\mathcal{A'}$ consists of three steps. In a first step, we compute and ensure all nodes know the number of nodes $n$, for example by computing a BFS tree, aggregating the number of nodes in each subtree, and having the root broadcast $n$ to all nodes. After which, in a second step, we compute an $(O(\log n),O(\log n))$-LDC decomposition of the communication graph $G$. Finally, in a third step, each cluster center $c_i$ (i.e., the root of the tree  of depth $O(\log n)$ spanning cluster $V_i$ in the decomposition) gathers all of its cluster nodes' input information using an upcast operation. 
More precisely, each node $v$ sends its input $in(v)$, broken up in $O(\log n)$ bit messages, to its parent in the cluster tree. (This includes the $O(d(v) \log n)$ bits describing its 1-hop neighborhood, or in other words, all ID pairs describing its $d(v)$ incident edges.) Whenever a node receives a message from its children, it simply transmits that message to its parent whenever possible.

\paragraph{Simulation.} The simulation part of $\mathcal{A}'$ is broken up into $\TA$ phases. Each phase $p \in [1,\TA]$ of $\mathcal{A'}$ is used to simulate one round (round $p$) of $\mathcal{A}$. Throughout most of these phases, the cluster centers take care of the simulation whereas non-center cluster nodes simply transmit information between cluster centers. However in the final phase, cluster centers inform their cluster's nodes of their outputs. 
Somewhat more formally, each phase $p \in [1,\TA]$ satisfies the following invariant: at the start of phase $p$, each cluster center $c_i$ knows the state of each of its cluster's node $v \in V_i$ at the start of round $p$ in the simulated algorithm $\mathcal{A}$. Moreover, at the end of the last phase $\TA$, all nodes know their state at the end of round $\TA$ in $\mathcal{A}$.

Now let us describe some given phase $p \in [1,\TA]$ of $\mathcal{A'}$. Phase $p$ simulates round $p<\TA$ of $\mathcal{A}$ using $O(n \log n)$ rounds, and round $p=\TA$ of $\mathcal{A}$ using $O(n \log n + Out)$ rounds. This is done in two steps, (1) send messages to neighboring clusters and (2) receive the messages from neighboring clusters, just as each round of the \congest{} model is also separated into these two steps (but where communication happens with neighboring nodes), with some computations happening before (1) and/or after (2).  
In fact, the phase starts with each cluster center $c_i$ performing the following series of local computations. (Recall that we mentioned previously that by this point, $c_i$ knows the state of each node $v \in V_i$ at the start of round $p$ of $\mathcal{A}$.) First, $c_i$ computes which nodes of $V_i$ would broadcast in round $p$ of $\mathcal{A}$, the set of which is denoted by $V_{i,p}$. For any node $v \in V_{i,p}$, let $m_{v,p}$ denote the message broadcasted by $v$ in round $p$ (with $v$ in the state computed by $c_i$), and let $E_v$ be the $O(\log n)$ edges of $F$ oriented from $v$ to some other node in a different cluster (where $F$ is the sparse inter-cluster communication directed edge set of the graph decomposition). Then, $c_i$ computes a set of node and message pairs $\mathcal{B}_{i,p} = \{(v, m_{v,p}) \mid v \in V_{i,p}\}$, and from that in turn computes a set of inter-cluster edges and message pairs $\mathcal{M}_{i,p}$, as follows. Start with $\mathcal{M}_{i,p}$ empty, and for each $(v, m_{v,p}) \in \mathcal{B}_{i,p}$, add the set $\{(e, m_{v,p}) \mid e \in E_v\}$ to $\mathcal{M}_{i,p}$.  

Once these local computations are done, the first step takes place. In it, each cluster center $c_i$ downcasts each message in $\mathcal{M}_{i,p}$ to the appropriate node in the cluster $V_i$. More precisely, for any edge and message pair $(e,\eta) \in \mathcal{M}_{i,p}$, let $u$ denote without loss of generality the endpoint of $e$ within $V_i$. Then, $c_i$ downcasts the edge and message pair $(e,\eta)$ to $u$, and upon receival, node $u$ sends the message $\eta$ to the other endpoint of $e$. Note that the downcast may incur significant congestion, in fact, $O(n \log n)$ congestion. Hence, all nodes run the first step for some $t_1 = O(n \log n)$ rounds before starting the second step simultaneously. (Here, $t_1$ depends on the analysis.)

In the second step, all nodes upcast any message received over an inter-cluster edge to the root of their cluster tree. (Similarly to the first step, this takes some $O(n \log n)$ rounds.) At the end of the upcast operation, each cluster center $c_i$ can and does simulate for each node $v \in V_i$ how its state would change in round $p$ of algorithm $\mathcal{A}$. (Recall that this depends on the messages $v$ receives in round $p$ of algorithm $\mathcal{A}$, and these have just been upcast to $c_i$ within algorithm $\mathcal{A'}$.) 

Most phases end at this point, but an additional step is executed for the final phase $p = \TA$. Indeed, each cluster center $c_i$ knows for each of its cluster's nodes, its output in $\mathcal{A}$. Then, $c_i$ downcasts these outputs to the respective nodes of cluster $V_i$. (This takes $O(Out)$ rounds, where recall that $Out$ denotes the size of the outputs over all nodes of $\mathcal{A}$.) Once that is done, the final phase terminates, and thus $\mathcal{A}'$ also.

\subsection{Analysis}
\label{subsec:MPXSimulationAnalysis}

We start by showing that $\mathcal{A'}$ correctly simulates $\mathcal{A}$, and in particular, that the cluster centers correctly simulate one round of $\mathcal{A}$ per phase of $\mathcal{A}'$.

\begin{lemma}
    At the end of any phase $p \in [1,\TA]$ of algorithm $\mathcal{A'}$, each cluster center $c_i$ knows the state of each of its cluster's node $v \in V_i$ at the end of round $p$ in the simulated algorithm $\mathcal{A}$.
\end{lemma}

\begin{proof}
    Let us show, by induction on phase $p \in [1,\TA]$, a slightly different statement: at the start of any phase $p \in [1,\TA]$ of algorithm $\mathcal{A'}$, each cluster center $c_i$ knows the state of each of its cluster's node $v \in V_i$ at the start of round $p$ in the simulated algorithm $\mathcal{A}$.
    
    The base case of phase $p = 1$ holds true due to cluster centers gathering their cluster nodes' inputs in the preprocessing part.
    Now, we assume the induction hypothesis holds for phase $p \in [1,\TA-1]$, and to prove the induction step we show the following claim: for any node $v \in V_i$, $c_i$ knows, by the end of phase $p$, all of the messages $v$ would have received in round $p$ of algorithm $\mathcal{A}$. 
    
    First, note that at the start of phase $p$, each cluster center $c_i$ correctly computes which nodes would broadcast in round $p$ of algorithm $\mathcal{A}$. After which, $c_i$ knows each message a node $v \in V_i$ would receive from a neighbor in the same cluster $V_i$. Now, we consider messages $v$ would receive from a neighbor $u$ in some different cluster $V_{i'}$. Note that $v$ might not receive such messages, since it could be that $(u,v) \notin F$. However, $u$'s message is correctly downcast from its cluster center $c_{i'}$ to $u$, then sent from $u$ to some node $w \in V_i$, and finally upcast from $w$ to $c_i$, such that $c_i$ knows $u$'s message by the end of phase $p$. (Note that $u$ transmits a single $O(\log n)$ bit message to $V_i$ in $\mathcal{A}$ because of the \bcongest{} nature of $\mathcal{A}$. Moreover, note that all messages are correctly downcast and upcast despite the congestion over cluster trees because we allow for $O(n \log n)$ rounds in every phase $p < \TA$, and $O(n \log n + Out)$ rounds in phase $p = \TA$.) Since $c_i$ additionally knows the IDs of both endpoints of all edges incident to any node in $V_i$, it also knows that the message sent by $u$ would be received by $v$ in round $p$ of algorithm $\mathcal{A}$. Thus, $c_i$ knows, by the end of phase $p$, all of the messages any node $v \in V_i$ would have received in round $p$ of algorithm $\mathcal{A}$. 

    Finally, the lemma is obtained by taking the above statement for $p = \TA$, and applying the claim one last time.
\end{proof}

Next, we consider the round and message complexities of $\mathcal{A'}$ (see Theorem \ref{thm:messageEfficientSimulation}). We start with the preprocessing part, where we recall that $In$ denotes the size over all nodes of the input of $\mathcal{A}$.

\begin{lemma}
\label{lem:preprocessing}
    The preprocessing steps take $O(In/ \log n + D \log n )$ time and $O(In+m \log n)$ messages.
\end{lemma}

\begin{proof}
    It is straightforward to see that the first step takes $O(D \log n)$ rounds and $O(m \log n)$ messages (deterministically, see \cite{kuttenjacm15}). As for the second step, by Lemma \ref{lem:LDCGraphDecomposition}, computing the LDC decomposition of $G$ takes $O(\log n)$ rounds and thus $O(m \log n)$ messages. Finally, the runtime and message complexity of the third step follows from Lemma \ref{lem:upcast} and the fact that the cluster trees have depth $O(\log n)$. 
\end{proof}

As for the simulation part, we recall that $Out$ denotes the (maximum possible) size over all nodes of the output of $\mathcal{A}$.

\begin{lemma}
\label{lem:simulation}
    The simulation part takes $O(Out + \TA \, n \log n)$ time and $O(Out \log n + \BA \log^2 n)$ messages.
\end{lemma}

\begin{proof}
    The round complexity follow directly from the algorithm description: all phases up to, but not including, $\TA$ take $O(n \log n)$ rounds, and phase $\TA$ takes $O(n \log n+ Out)$ rounds. Hence in total, the round complexity is  $O(Out + \TA \, n \log n)$. 
    
    It remains to bound the message complexity. Let us consider some arbitrary phase $p \in [1,\TA]$. We bound the message complexities separately for the different steps of phase $p$, starting with the first step. Let $U_{p}$ denote the set of nodes (over all clusters) broadcasting a message in round $p$ of the simulated algorithm $\mathcal{A}$. Each broadcasting node induces at most $O(\log n)$ downcast messages since there are at most $O(\log n)$ edges in $F$ oriented from $v$ to some node in a different cluster. Thus at most $O(|U_p| \log n)$ messages are downcast over all cluster trees. Since the cluster trees have depth $O(\log n)$, then by Lemma \ref{lem:downcast} the downcast operation uses at most $O(|U_p| \log^2 n)$ messages. 
    Additionally, the first step ends with the communication over (directed) inter-cluster edges, which uses $O(|U_p| \log n)$ messages. Altogether, the first step has $O(|U_p| \log^2 n)$ message complexity. 
 
    For the second step, it is clear that at most $O(|U_p| \log n)$ messages (each of $O(\log n)$ bits) are received via inter-cluster edges, and thus upcast, over all clusters. (Note that some nodes may receive more than $O(\log n)$ messages in the second step.) Then, the $O(\log n)$ depth of the cluster trees and Lemma \ref{lem:upcast} imply that the upcast (and thus second step) uses at most $O(|U_p| \log^2 n)$ messages.
    Finally, for the third step (executed by phase $\TA$ only), at most $Out$ messages are downcast over all clusters. This implies that the third step uses at most $O(Out \log n)$ messages (by Lemma \ref{lem:downcast}).

    To sum up, the simulation part takes $O(Out + \TA \, n \log n)$ rounds and uses at most $O(Out \log n + \sum_{p \in [1,\TA]} |U_p| \log^2 n)$ messages. As $\sum_{p \in [1,\TA]} |U_p|$ corresponds to the broadcast complexity $\BA$ of $\mathcal{A}$, the lemma statement follows.
\end{proof}

Finally, the round complexity and message complexity of $\mathcal{A}'$ (see Theorem \ref{thm:messageEfficientSimulation}) follows by applying Lemmas \ref{lem:preprocessing} and \ref{lem:simulation} (and the fact that $In \geq m \log n$). Furthermore, we point out that the current simulation provides a Monte Carlo randomized algorithm $\mathcal{A'}$ even if the original algorithm $\mathcal{A}$ is deterministic due to the randomized algorithm for computing an $(O(\log n),O(\log n))$-LDC decomposition presented in Section \ref{subsec:LDC}. This can be improved by replacing the randomized LDC decomposition algorithm with a deterministic one.

\subsection{Summary of the Applications of this Simulation}
\label{subsec:simulationApplications1}

As a first warm-up application, we consider the APSP problem on weighted graphs, where each edge is labeled with a weight chosen from a range that is polynomial in $n$.
The goal is for each node to output its distance to all other nodes. 
The work of \cite{bernstein2021distributed} presents a randomized Las Vegas algorithm that computes the (exact) APSP problem in $\tilde{O}(n)$ rounds of the \bcongest{} model.  
An immediate application of Corollary~\ref{cor:messageEfficientSimulation} yields the following:

\MessageOptimalAPSP*

In the paper, we mostly focus on the APSP problem, and in particular on its unweighted version. However, the simulation given in this section is especially general, and leads to message-efficient algorithms for other problems. In the Appendix, we describe how to obtain such message-efficient algorithms for (1) exact bipartite maximum matching, and (2) neighborhood covers (in particular, those with small $o(\log n)$ radius). More concretely, we show the following two statements---Corollary \ref{cor:bipartiteMaxM} and \ref{cor:neighborhoodCover}---in Section \ref{sec:apps} of the Appendix.

\begin{restatable}{corollary}{bipartiteMaxM} \label{cor:bipartiteMaxM}
There is an algorithm in the \congest{} model with $\tilde{O}(n^2)$ message complexity and $\tilde{O}(n^2)$ round complexity that computes an exact maximum matching in bipartite graphs.
\end{restatable}

For the below statement, we quickly state the standard definition of a \emph{$(k,W)$-sparse neighborhood cover}, where $k$ and $W$ are any positive integers. The cover is a collection of trees $\mathcal{C}$ such that (1) every tree in $\mathcal{C}$ has depth at most $O(W \cdot k)$, (2) each vertex $v$ appears in at most $\tilde O(k \cdot n^{1/k})$ different trees, and (3) there exists some tree in $\mathcal{C}$ that contains the entire $W$-neighborhood of $v$.

\begin{restatable}{corollary}{neighborhoodCover} \label{cor:neighborhoodCover}
For positive integers $k$ and $W$, there is a $\congest$ algorithm that, with high probability, computes a $(k,W)$-sparse neighborhood cover in $\tilde O(n^2)$ rounds and with $\tilde O(n^2)$ message complexity.
\end{restatable}

\section{Message-Time Trade-offs for APSP}
\label{sec:messageRoundTradeoffs}
In Section \ref{sec:simulation}, we show that message-optimality---more precisely, $\tilde{O}(n^2)$ message complexity---is achievable for weighted APSP, and several other problems. However, this comes at the cost of a high and clearly non-optimal round complexity of $\tilde{O}(n^2)$. On the other hand, round-optimal (i.e, with $\tilde{O}(n)$ runtime) algorithms for weighted APSP exist (e.g., \cite{bernstein2021distributed}) but they are not message-optimal---they have $\tilde{O}(n^3)$ message complexity.
In this section, we provide a distributed \congest{} algorithm achieving a message-time trade-off for unweighted APSP that provides a natural, linear trade-off between the above mentioned two points, i.e., ($\tilde{O}(n^2)$ rounds, $\tilde{O}(n^2)$ messages) at one end and ($\tilde{O}(n)$ rounds, $\tilde{O}(n^3)$ messages) at the other:

\APSPTradeoff*

Recall that we obtained the results in Section \ref{sec:simulation}  by presenting a general simulation theorem for \bcongest{} algorithms.
Here, Theorem \ref{thm:APSPTradeoff} is also obtained via a simulation, but this simulation applies to \bcongest{} algorithms that satisfy an additional constraint: that is, the algorithm should be \textit{aggregation-based}, in the sense that there is an aggregation function that, intuitively speaking, allows us to replace any collection of messages addressed to the same destination by equivalent messages that have a length of only $\tilde O(1)$ bits:

\begin{definition}[Aggregation-based Algorithm]
\label{def:aggregationBasedAlgorithm} 
Consider a \bcongest{} algorithm $\mathcal{A}$. For an arbitrary node $v$ and an arbitrary round $r$, let $\mathbb{M}_{v, r}$ denote the set of  possible messages that $v$ may receive in round $r$, let $\mathcal{P}(\mathbb{M}_{v, r})$ denote its power set, and let $f_{v, r}$ denote the local function computed by node $v$ at the end of round $r$.
Algorithm $\mathcal{A}$ is an \textit{aggregation-based} algorithm if, for all nodes $v$ and rounds $r$, there exists a function $agg_{v,r} : \mathcal{P}(\mathbb{M}_{v, r}) \to \mathcal{P}(\mathbb{M}_{v, r})$ such that, for any subset $M \subseteq\mathbb{M}_{v, r}$: 
\begin{itemize} 
 \item $agg_{v,r}(M)$ can be represented in $\tilde{O}(1)$ bits,\footnote{Note that $agg_{v,r}(M)$ returns a subset of the messages of $M$ rather than just a single message.} and
\item for any partition
$(M_1',\dots,M_k')$ of $M$, it holds that
$$\textstyle f_{v,r}(\textsf{state}_v, M) =f_{v,r}(\textsf{state}_v, \bigcup_{i=1}^k agg_{v, r}(M_i)).$$
\end{itemize}
\end{definition}

To take full advantage of the message improvements given by the new simulation, we consider \bcongest{} algorithms that we call {$\ell$-decomposable}:  
\begin{definition}
An \textit{$\ell$-decomposable} algorithm $\mathcal{A}$ consists of $\ell$ completely independent \bcongest{} algorithms $\mathcal{A}_1, \mathcal{A}_2, \ldots, \mathcal{A}_\ell$ that we call \textit{components}. 
Every node $v$ locally computes its final output from its $\ell$ outputs, one from each of the $\ell$ components.
\end{definition}

We now give an overview of the underlying techniques of our simulation. The basis for our simulation are the cluster hierarchies defined by Baswana and Sen \cite{BaswanaSenRSA2007} as part of their classical randomized spanner algorithm. 
Baswana and Sen present an algorithm that, for any input graph $G = (V, E)$ and any integer $\kappa \ge 1$, constructs $(2\kappa-1)$-spanner $H = (V, E_H)$ of $G$ with $O(n^{1+1/\kappa})$ edges.
While the final output of this algorithm is a spanning subgraph $H$ of $G$, a byproduct of this algorithm is a $(\kappa+1)$-level hierarchy of clusters (to be defined precisely in the next subsection). It is this Baswana-Sen cluster hierarchy that we utilize for our simulation. 
However, using a single Baswana-Sen cluster hierarchy leads to excessive congestion on a few edges and very little congestion on many edges. 
A key idea to bypass this congestion bottleneck is to use not one, but a small number $\zeta$ of independently constructed Baswana-Sen cluster hierarchies. 
Then the set of $\ell$ components of an $\ell$-decomposable algorithm $\mathcal{A}$ can be partitioned into $\zeta$ equal-sized batches and each batch can be assigned a different Baswana-Sen cluster hierarchy. 
This simple and natural idea plays a crucial role in reducing the maximum congestion of an edge and allows us to appeal to the congestion plus dilation framework (see Section \ref{section:congestionPlusDilation}) to bound the round complexity of our simulation.
One additional requirement of our simulation is that each proper subtree of a cluster in the Baswana-Sen cluster hierarchy needs to have a relatively small number of nodes. 
To ensure this we need an additional pruning phase.
The final product is an \textit{ensemble of pruned Baswana-Sen cluster hierarchies} that we then use in our simulation.

\subsection{Ensemble of Pruned Baswana-Sen Cluster Hierarchies}
\label{subsec:BaswanaSen}

Let $\eps \in [\frac{1}{\Theta(\log n)},1]$ be some input parameter. 
Let $\kappa := \keps{}$.
First, we compute a sequence of $\kappa+1$ subsets of vertices $S_0,S_1,\ldots,S_{\kappa}$ via random sampling as follows. We define $S_0 = V$, $S_{\kappa} = \emptyset$, and for any $i \in [1,\kappa-1]$, $S_i$ is obtained by sampling nodes independently from $S_{i-1}$ with probability $n^{-\eps}$. 
These subsets will form centers of clusters defined below.

Second, for each level $i = 0, 1, \ldots, \kappa$, we compute three objects:
\begin{itemize}
    \item \textit{clustering} $\mathcal{C}_i$; each clustering $\mathcal{C}_i$ being a collection of vertex-disjoint clusters, 
    \item \textit{low-degree vertex subset} $L_i$, and 
    \item \textit{inter-cluster communication edge subset} $F_i$. 
\end{itemize}
The sequence $(\mathcal{C}_i, L_i, F_i)_{i=0}^{\kappa}$ consisting of $\kappa+1$ levels defines a \textit{Baswana-Sen cluster hierarchy.} We now describe its construction in detail.

\begin{itemize}
    \item \textbf{Bottom level:} At the bottom level, i.e., level $i = 0$, the clustering $\mathcal{C}_0 := \{\{v\} \mid v \in V\}$. Furthermore, the low-degree vertex subset $L_0$ and the inter-cluster communication edge subset $F_0$ are both empty. 
    For any cluster $C = \{v\}$ in $\mathcal{C}_0$, we designate $v$ as the \textit{center} of cluster $C$.
    Note that this means that the vertices in $S_0$ form the centers of clusters in $\mathcal{C}_0$.
    Each cluster $C \in \mathcal{C}_0$ can be viewed as (trivial) rooted tree.
    We use $V_0$ to denote the set of vertices that belong to clusters in $\mathcal{C}_0$. Clearly, $V_0 = V$. More generally, we will use $V_i$ to denote the vertices that belong to clusters in $\mathcal{C}_i$.
    
    \item \textbf{Levels $i = 1, 2, \ldots, \kappa-1$:} We obtain the clustering $\mathcal{C}_{i+1}$, low-degree vertex subset $L_{i+1}$, and inter-cluster communication edge subset $F_{i+1}$ for level $i+1$, $0 \le i \le \kappa-2$, from level $i$ as follows. Let $\mathcal{R}_i \subseteq \mathcal{C}_i$ denote the subset of clusters whose centers are in $S_{i+1}$; we call these \textit{sampled $i$-clusters}. Consider each non-sampled $i$-cluster $C$, i.e., a cluster $C \in \mathcal{C}_i \setminus \mathcal{R}_i$ and each node $v \in C$. There are two cases depending on the neighborhood of $v$.
    \begin{itemize}
        \item[(i)] \textbf{$v$ has a neighbor that belongs to a sampled $i$-cluster.} Then $v$ joins an arbitrary sampled $i$-cluster in its neighborhood, via an edge to some arbitrary neighbor in that cluster. The edge along which $v$ joins the neighboring sampled $i$-cluster is called a \textit{cluster edge} and is denoted by $parent(v, i+1)$. This process results in the growth of clusters in $\mathcal{R}_i$; this set of sampled $i$-clusters, along with their newly-joined nodes, forms the new clustering $\mathcal{C}_{i+1}$. 
        Note that each cluster in $\mathcal{C}_{i+1}$ can be viewed as a rooted tree, rooted at a center in $S_{i+1}$.     
        
        \item[(ii)] \textbf{$v$ does not have a neighbor that belongs to a sampled $i$-cluster.} Then node $v$ is added to $L_{i+1}$.
        Furthermore, for each neighboring cluster in $\mathcal{C}_i$, one (undirected) edge from $v$ to some (arbitrary) vertex in that cluster is added to $F_{i+1}$.
        Note that the vertex subsets $L_{i+1}$ and $V_{i+1}$ form a partition of $V_i$.
    \end{itemize}
    \item \textbf{Top level:} At the top level, $C_{\kappa} = \emptyset$ and thus $V_{\kappa} = \emptyset$. We set $L_{\kappa} = V_{\kappa-1}$ and for each $v \in L_\kappa$ and each neighboring cluster $C \in \mathcal{C}_{\kappa-1}$, we add one (undirected) edge from $v$ to some vertex in $C$ to $F_{\kappa}$.
\end{itemize}
The set of all \textit{inter-cluster communication edges} $F$ is simply $\cup_{i=0}^{\kappa} F_i$. This completes the construction of the Baswana-Sen cluster hierarchy. The following theorem states well-known properties of the Baswana-Sen cluster hierarchy \cite{BaswanaSenRSA2007}. 

\begin{theorem}
\label{theorem:BaswanaSenProperties}
The Baswana-Sen cluster hierarchy $(\mathcal{C}_i, L_i, F_i)_{i=0}^{\kappa}$ has the following properties:
\begin{itemize}
\item[(a)] For any level $i \in [0,\kappa-1]$, the clustering $\mathcal{C}_i$ is a collection of disjoint clusters of $V_i$ with (strong) radius $i$.
\item[(b)] The following statement holds with high probability: for any level $i \in [1,\kappa]$, for any node $v \in L_i$, there are at most 
$O(n^\eps \log n)$ edges in $F_i$ incident on $v$, each connecting $v$ to a distinct cluster in $\mathcal{C}_{i-1}$.
\item[(c)] Consider an edge $(u, v) \in E$ and suppose that $u \in L_i$ and $v \in L_j$ for $i \le j$. Then either (1) some cluster in $\mathcal{C}_{i-1}$ contains both $u$ and $v$, or (2) there exists some inter-cluster communication edge $e = (u, w) \in F_i$ and some cluster in $\mathcal{C}_{i-1}$ containing both $w$ and $v$. 
\end{itemize}
\end{theorem}

The construction of the Baswana-Sen cluster hierarchy described above can be implemented in a straightforward manner in the \congest{} model, leading to the following well-known theorem.

\begin{theorem}
\label{theorem:BaswanaSenRunningTime}
Let $G = (V,E)$ be an unweighted graph. Then, there exists a \congest{} algorithm running in $O(\kappa)$ rounds and using $O(\kappa \cdot m)$ messages for computing a
Baswana-Sen cluster hierarchy with $\kappa + 1$ levels. 
\end{theorem}

\textbf{Pruning the clusters.}
In order to use the Baswana-Sen cluster hierarchy for efficient simulation, we need a pruning step to avoid clusters with large subtrees. In fact, clusters with large subtrees suffer from high maximum edge congestion in the simulation. More concretely, we want to avoid having any clusters in which any proper subtree has $\omega(n^{1-\eps})$ nodes. 

We now describe how the pruned Baswana-Sen hierarchy $(\mathcal{C}^*_i, L_i, F^*_i)_{i=0}^{\kappa}$ is obtained from the Baswana-Sen hierarchy $(\mathcal{C}_i, L_i, F_i)_{i=0}^{\kappa}$. 
This pruning step involves (i) pruning each cluster at each level, if necessary and (ii) adding inter-cluster communication edges. 
\begin{itemize}
    \item \textbf{Pruning clusters:} Consider each level $i \in [1,\kappa-1]$ (for levels 0 and $\kappa$, pruning is not required) and each cluster $C \in \mathcal{C}_i$.
    If every proper subtree in $C$ contains fewer than $n^{1-\eps}$ nodes, then there is no need to prune $C$. Otherwise, we repeatedly look for the deepest node, say $u^*$, such that the subtree rooted at $u^*$ contains at least $n^{1-\eps}$ nodes.
    We then split off the subtree rooted in $u^*$ into its own cluster. Note that this can happen at most $O( n^{\eps})$ times at a particular level $i$, by a simple counting argument over the $n$ nodes, thereby adding at most $O( n^{\eps})$ clusters to the existing clustering $\mathcal{C}_i$. Once there exists no such node for all clusters in a clustering $\mathcal{C}_i$, we have the new level-$i$ clustering $\mathcal{C}^*_i$. 
    Note that while this pruning procedure is described as a seemingly sequential algorithm, all clusters within a clustering $\mathcal{C}_i$ at level $i$ can be processed in parallel using upcast, leading to an algorithm that runs in $O(\kappa^2)$ rounds, using $O(\kappa \cdot n)$ messages.  
    \item \textbf{Adding inter-cluster communication edges:} Consider some level $i \in [1,\kappa]$. (For level 0, $F^*_0$ is empty, similarly to $F_0$.) Then, for each node $v \in L_{i}$ and for each neighboring cluster $C \in \mathcal{C}^*_{i-1}$, we add to $F^*_i$ one edge from $v$ to some node $C$. Because $v$ has only $O(n^{\eps})$ neighboring clusters in $\mathcal{C}_{i-1}$, and only $O(n^{\eps})$ new clusters are created per clustering, there are only $O(n^{\eps})$ edges incident to $v$ in $F^*_i$.
\end{itemize}
We call the result \textit{pruned Baswana-Sen cluster hierarchy}. 

The following two corollaries are versions of Theorems \ref{theorem:BaswanaSenProperties} and \ref{theorem:BaswanaSenRunningTime} that apply to pruned Baswana-Sen cluster hierarchies.
\begin{corollary}
\label{thm:PrunedBaswanaSenRunningProperties}
Properties (a), (b), and (c) in Theorem \ref{theorem:BaswanaSenProperties} also hold for a pruned Baswana-Sen cluster hierarchy. In addition, each proper subtree of a cluster in a pruned Baswana-Sen cluster hierarchy contains at most $O(n^{1-\epsilon})$ nodes. 
\end{corollary}

\begin{corollary}
\label{theorem:PrunedBaswanaSenRunningTime}
Let $G = (V,E)$ be an unweighted graph. Then, there exists a \congest{} algorithm running in $O(\kappa^2)$ rounds and using $O(\kappa \cdot m)$ messages for computing a
pruned Baswana-Sen cluster hierarchy with $\kappa + 1$ levels. 
\end{corollary}

\textbf{Ensemble of pruned Baswana-Sen cluster hierarchies.} Instead of using one pruned Baswana-Sen cluster hierarchy to simulate our algorithm, we use multiple pruned Baswana-Sen cluster hierarchies with each cluster hierarchy responsible for simulating some number of components of an $\ell$-decomposable algorithm. The motivation for doing this is to ``smooth'' out congestion on cluster edges, i.e., instead of a few cluster edges (of a particular cluster hierarchy) having extremely high congestion, we end up with many more cluster edges, each having lower congestion. 
Recall that by ``cluster edge'' we mean an edge of a rooted tree that forms a cluster.
This idea crucially relies on the following lemma that shows that the likelihood of an edge ending up as a cluster edge in a pruned Baswana-Sen cluster hierarchy is quite small. 
\begin{lemma}
\label{lemma:clusterEdgesRare}
For any edge $e$ of the input graph $G$, $e$ is a cluster edge of a pruned Baswana-Sen cluster hierarchy with probability $O(\kappa \cdot n^{-\eps})$.
\end{lemma}
\begin{proof}
Consider any edge $e = (u,v)$ of $G$, as well as any level $i \in [0,\kappa-2]$. We show that $e$ is a cluster edge in $\mathcal{C}_{i+1}$ with probability $O(n^{-\eps})$. (Note that $\mathcal{C}_{0} = \{\{v\} \mid v \in V\}$ and $\mathcal{C}_{\kappa} = \emptyset$, thus edges cannot be cluster edges in levels 0 and $\kappa$.) To do so, we fix all random choices in all levels up to $i$ except those for $S_i$, by the principle of deferred decisions. Clearly, $e$ is a cluster edge in $\mathcal{C}_{i+1}$ only if either $parent(u,i+1) = e$ or $parent(v,i+1) = e$. The first case only happens if the cluster in which $v$ belongs (if any, and if $u$ and $v$ are in different clusters in $C_i$) is a sampled $i$-cluster (and the second case is symmetrical). Since the probability that the cluster containing $v$ (if any) is a sampled $i$-cluster is at most $n^{-\eps}$, the probability that $parent(u,i+1) = e$ is also at most $n^{-\eps}$. 

Finally, a union bound over the two cases and the $O(\kappa)$ levels proves the statement holds for the Baswana-Sen cluster hierarchy. But since no cluster edge is added (only removed) during the pruning step, the statement also holds for the pruned Baswana-Sen cluster hierarchy.
\end{proof}

Let $\mathcal{A}$ be an $\ell$-decomposable \bcongest{} algorithm and let $\mathcal{A}_1, \mathcal{A}_2, \ldots, \mathcal{A}_\ell$ be its $\ell$ components. Suppose that each algorithm $\mathcal{A}_j$ takes as input a pruned Baswana-Sen cluster hierarchy.
In other words, each algorithm $\mathcal{A}_j$ uses a given pruned Baswana-Sen cluster hierarchy $(\mathcal{C}_i, L_i, F_i)_{i=0}^\kappa$ in some way in its execution. Let $\congestion_j$ denote the maximum number of messages that pass over a cluster edge, in the worst case, during the execution of $\mathcal{A}_j$. Here, the ``worst case'' is over all possible inputs, including all possible pruned Baswana-Sen cluster hierarchies. Let $\congestion = \max_j \congestion_j$. 
Now suppose that we execute all $\ell$ algorithms together, using a single pruned Baswana-Sen cluster hierarchy $\mathcal{H}=(\mathcal{C}_i, L_i, F_i)_{i=0}^\kappa$ as input to all $\ell$ algorithms. 
Then it is possible that in the worst case there is a cluster edge of $\mathcal{H}$ over which $\ell \cdot \congestion$ messages pass.
This happens if the same cluster edge $e$ of $\mathcal{H}$ sees the maximum number of messages passing over it, in all of the $\ell$ algorithms.
We now show that we can spread out this congestion across many different edges by simply having the algorithms use an ensemble of pruned Baswana-Sen cluster hierarchies, instead of a single hierarchy.

\begin{lemma}[Congestion Smoothing Lemma]
\label{lem:congestionSmoothing}
Let $\zeta = \lceil n^{\eps} \rceil$ and consider an ensemble of $\zeta$ independently constructed, pruned Baswana-Sen cluster hierarchies $\{\mathcal{H}_1, \mathcal{H}_2, \ldots, \mathcal{H}_\zeta\}$.
Arbitrarily partition the set of $\ell$ algorithms that are the components of $\mathcal{A}$ into $\zeta$ equal-sized batches $B_1, B_2, \ldots, B_{\zeta}$ and provide each pruned Baswana-Sen cluster hierarchy $\mathcal{H}_j$ as input to all the algorithms in batch $B_j$.
If we execute the components $\mathcal{A}_1, \mathcal{A}_2, \ldots, \mathcal{A}_\ell$ with this input distribution, then with high probability, the maximum number of messages that pass
over any cluster edge over the course of algorithm $\mathcal{A}$ is $\tilde{O}\left(\frac{\congestion \cdot \ell}{\zeta}\right)$.
\end{lemma}
\begin{proof}
First, we consider any edge $e$ and show that, with high probability, it is a cluster edge in at most $O(\log n)$ different cluster hierarchies. We start by recalling that for any $\mathcal{H}_j$, edge $e$ is a cluster edge of $\mathcal{H}_j$ with probability at most $O(\kappa \cdot n^{-\eps})$, by Lemma \ref{lemma:clusterEdgesRare}. By linearity of expectation (over the $\zeta = \lceil n^{\eps} \rceil $ cluster hierarchies), edge $e$ is the tree edge of at most $O(\kappa \cdot n^{-\eps} \cdot \zeta) = O(\log n)$ different clusters in expectation. Now, note that the random choices are independent between the cluster hierarchy computations. Hence, a simple Chernoff bound shows that $e$ is a cluster edge in at most $O(\log n)$ different cluster hierarchies with high probability. 

Let us condition on this high probability event and let $\mathbb{H}_e$ denote the $O(\log n)$ different cluster hierarchies that $e$ is a cluster edge in.
The cluster hierarchies in $\mathbb{H}_e$ are assigned to a total of $O(\ell/\zeta \cdot \log n)$ components. Therefore, with high probability, the total congestion on $e$ is 
$\tilde O(\frac{\congestion \cdot \ell}{\zeta}).$
\end{proof}

Although the above lemma statement is given for any $\ell, \zeta$, note that when $\ell = \zeta$ the resulting statement remains interesting. Indeed, in that case, the maximum congestion that passes over any cluster edge over the course of algorithm $\mathcal{A}$ is $\tilde{O}(\log n)$.

\subsection{Message-Time Trade-off Simulations}
\label{subsec:simulationTradeoff}

In this section, we provide a simulation that works for any \bcongest{} aggregation-based algorithm $\mathcal{A}$ with a known upper bound $\TA$ on the round complexity---see Definition \ref{def:aggregationBasedAlgorithm} for a formal definition of aggregation-based algorithms, but these are, roughly speaking, algorithms in which all nodes locally update their state every round using some \textit{aggregate} function such as $\min, \max, \mathrm{average}$, etc. Algorithm $\mathcal{A}$ may use up to $O(\TA \cdot n^2)$ messages (say, in dense graphs). We present a simulation that, informally, allows to obtain a \congest{} algorithm $\mathcal{A}'$ simulating $\mathcal{A}$ (i.e., producing the same output) sending less messages at the cost of using more rounds---see Theorem \ref{thm:tradeoffSimulation1} below.

\begin{restatable}{theorem}{simulationGeneral}
\label{thm:tradeoffSimulation1}
    For any $\eps \in [\frac{1}{\Theta(\log n)},1]$ and a pruned Baswana-Sen cluster hierarchy $\mathcal{H}$ of parameter $\eps$, any \bcongest{} aggregation-based algorithm $\mathcal{A}$ with (known upper bound on the) runtime $\TA$ can be converted into a randomized (Monte Carlo) \congest{} algorithm that:
    \begin{itemize}
        \item simulates $\mathcal{A}$ correctly (w.h.p.),
        \item runs in $\tilde{O}(\TA  \cdot n + n^{2-\eps})$ rounds,
        \item sends $\tilde{O}(\TA \cdot n^{1+\eps} + m)$ messages (w.h.p.),
        \item incurs $\tilde{O}(\TA)$ maximum edge congestion over non cluster edges of $\mathcal{H}$.
    \end{itemize}
\end{restatable}

Moreover, we adapt the above simulation to be faster when $\eps \in [1/2,1]$ (see Theorem \ref{thm:tradeoffSimulation2} below).

\begin{restatable}{theorem}{simulationFaster}
\label{thm:tradeoffSimulation2}
    For any $\eps \in [1/2,1]$ and a pruned Baswana-Sen cluster hierarchy $\mathcal{H}$ of parameter $\eps$, any \bcongest{} aggregation-based algorithm $\mathcal{A}$ with (known upper bound on the) runtime $\TA$ can be converted into a randomized (Monte Carlo) \congest{} algorithm that:
    \begin{itemize}
        \item simulates $\mathcal{A}$ correctly with high probability,
        \item takes $\tilde{O}(\TA  \cdot n^{1-\eps} + n)$ runtime,
        \item sends $\tilde{O}(\TA \cdot n^{1+\eps} + m)$ messages with high probability,
        \item incurs $\tilde{O}(\TA)$ maximum edge congestion over non cluster edges of $\mathcal{H}$.
    \end{itemize}
\end{restatable}

\subsubsection{General Simulation for $\eps \in [\frac{1}{\Theta(\log n)},1]$:}
\label{sec:general}
Our first, general simulation---in that it applies to a wider range of $\eps$ values---takes as input a pruned Baswana-Sen cluster hierarchy (with parameter $\eps$). Then, Algorithm $\mathcal{A}'$ can be decomposed into two parts: a preprocessing part and a simulation part. In the preprocessing part, nodes aggregate information over each of the clusters of the pruned Baswana-Sen cluster hierarchy. In the simulation part, we simulate the execution of $\mathcal{A}$ using the cluster hierarchy.
Referring back to Theorem \ref{thm:tradeoffSimulation1}, the $\tilde{O}(n^{2-\eps})$ term in the running time and the $\tilde{O}(m)$ term in the message complexity are due to preprocessing, whereas the $\tilde{O}(\TA\cdot n)$ term in the running time and the $\tilde{O}(\TA \cdot n^{1+\eps})$ term in the message complexity are due to the round-by-round simulation of $\mathcal{A}$.

\paragraph{Preprocessing.} The preprocessing part of $\mathcal{A}'$ consists of two steps. 
\begin{enumerate}
    \item Elect a leader, compute a BFS tree rooted in that leader, aggregate the number of nodes $n$ and finally broadcast $n$ to all nodes.
    \item Iteratively over the $\kappa+1$ levels of the pruned Baswana-Sen cluster hierarchy, do the following: For each cluster $C$ in a given level, its cluster center gathers the information of all edges incident to nodes in $C$. More precisely, each node in $C$ upcasts $O(d(v) \log n)$ bits describing its 1-hop neighborhood in the communication graph $G$ and in the pruned Baswana-Sen cluster hierarchy, broken up in $O(\log n)$ bit messages, to its parent in the cluster tree of $C$. Indeed, $O(d(v) (\log n+\kappa)) = O(d(v) \log n)$ bits can describe all ID pairs corresponding to the $d(v)$ incident edges to $v$, as well as the $\kappa = O(\log n)$ bits per such pair to indicate whether the corresponding edge is an inter-communication edge of a given level of the hierarchy.
\end{enumerate}

\paragraph{Simulation.} All nodes start the simulation (i.e., $\mathcal{A}'$) with the same initial state as in $\mathcal{A}$ (and some additional information obtained in the preprocessing). After which, any phase $p \in [1,\TA]$ of the simulation uses $\tilde{O}(n)$ rounds to simulate round $p$ of $\mathcal{A}$ (where the precise runtime of each phase depends on the analysis).
More precisely, if we let $B_p$ be the set of nodes broadcasting in round $p$ of $\mathcal{A}$ (given the above mentioned initial state), and among these we let $B_p(u)$ denote the neighbors of any node $u \in V$ contained in $B_p$, then each phase of $\mathcal{A}'$ simulates the communication of the corresponding round of $\mathcal{A}$ in the aggregate sense: that is, for any phase $p \in [1,\TA]$ and node $u \in V$, during phase $p$ node $u$ receives the aggregate of all messages sent by (nodes in) $B_p(u)$ (or in other words, function $agg_{u,r}$ applied to all messages sent by $B_p(u)$).

Each phase of $\mathcal{A}'$ simulates a round of $\mathcal{A}$ through the following three \textit{send, receive and compute} steps, described for some arbitrary node $v$:
\begin{enumerate}
    \item \textbf{Send:} If $v$ broadcasts a message $m_v$ in round $p$ of $\mathcal{A}$, then in this step (of phase $p$) node $v$ executes the \textit{indirect send and direct send} sub-steps that ensure that for each neighbor $u$ of $v$ and cluster $C$ containing $v$, either:
    \begin{enumerate}[label=(\roman*)]
        \item \textbf{Indirect Send:} The message $m_v$ is sent to some cluster containing $u$,
        \item \textbf{Direct (Aggregate) Send:} Or an aggregate (of at most $\tilde{O}(1)$ bits) of all messages from (nodes in) $B_p(u) \cap C$ is sent directly to $u$, 
    \end{enumerate}
    \item \textbf{(Aggregate) Receive:} For any cluster $C$ containing $v$ in the hierarchy, $v$ receives an aggregate of all messages sent to $C$ in the indirect send sub-step by some node in $B_p(v)$.
    \item \textbf{Compute:} Finally, $v$ locally computes the aggregate of all the messages received in the send---more precisely, in the direct send sub-step---and receive step of this phase (i.e., via its its incident inter-communication edges, and the up to $\kappa$ clusters $v$ can belong to) and updates its state for phase $p+1$ accordingly---just as $v$ would update its state for round $p+1$ after receiving this aggregate in round $p$ of $\mathcal{A}$. 
\end{enumerate}
When the send and receive steps are executed correctly---we give their implementation below---and $\mathcal{A}$ is an aggregation-based distributed algorithm, it is straightforward to see the simulation is correct; in particular, the output of each node in $\mathcal{A}'$ is exactly the output of that same node in $\mathcal{A}$. Note that the send step is separated into two sub-steps due the asymmetric nature of property (c) of the cluster hierarchy (see Theorem \ref{theorem:BaswanaSenProperties}). More precisely, for any two neighbors $u,v$, there is either (i) an inter-communication edge incident to $v$ and going to a cluster containing $u$, or (ii) an inter-communication edge incident to $u$ and going to some node in $C$ (but not necessarily $v$). Then, messages are sent between the cluster of $v$ and $u$ in the first case within the indirect send sub-step, and in the second case within the direct send sub-step.

Next, we describe the implementation of the send step of some phase $p$ (in $\mathcal{A}'$). Consider some node $v$ that broadcasts some message $m_v$ in round $p$ of $\mathcal{A}$. Then, in the send step of phase $p$, $v$ executes the two sub-steps sequentially:
\begin{enumerate}[label=(\roman*)]
    \item \textbf{Indirect Send:} Node $v$ sends $m_v$, appended with its ID, over any incident inter-communication edges of the pruned Baswana-Sen cluster hierarchy.
    \item \textbf{Direct (Aggregate) Send:} First, $v$ upcasts $m_v$ over all cluster trees it belongs to in the cluster hierarchy. Once all upcasts are done, let $C$ be any cluster, $\mathcal{M}(C)$ be all messages received by the center of $C$ (during the upcast) and $R(C)$ be the set of nodes outside $C$, but with a neighbor in $C$ and an inter-communication edge to $C$ (where the latter can be locally computed by the center of $C$ from the knowledge obtained during the preprocessing). Then, for any node $u \in R(C)$, the center of $C$ computes the aggregate of all messages in $\mathcal{M}(C)$ originating from nodes in $B_p(u)$. (Note that each such aggregate may consist of at most $\tilde{O}(1)$ bits, and sending this information is done via at most $\tilde{O}(1)$ bits, the collection of which we call an \textit{aggregate packet}.) After which, the center of $C$ downcasts, simultaneously for all nodes $u \in R(C)$, the corresponding computed aggregate packet to the endpoint in $C$ of (one of) the inter-communication edge incident to $u$, after appending each aggregate packet with the ID of $u$. (Note that the same node in $C$ may be the endpoint of many such inter-communication edges.) Finally, any aggregate packet received by some node $w$ (after the downcast terminates) is sent by $w$ over the corresponding inter-communication edge (i.e., where the other endpoint is the node with the ID appended to the received message).
\end{enumerate}

Now, we describe the implementation of the receive step of some phase $p$ (in $\mathcal{A}'$). Consider any node $v$ that receives some message in the indirect send sub-step of that same phase---importantly, these cannot be messages from an aggregate packet. Then, $v$ upcasts any such message over all cluster trees it belongs to. Moreover, all broadcasting nodes in $B_p$ also upcast their message over all cluster trees they belong to. Once all upcasts are done, let $C$ be any cluster and $\mathcal{M}(C)'$ be all messages received by the center of $C$ (during the upcast). Then, for any node $u \in C$, the center of $C$ computes the aggregate of all messages in $\mathcal{M}(C)'$ originating from nodes in $B_p(v)$. (Note that the cluster center can compute which messages in $\mathcal{M}(C)'$ originate from nodes in $B_p(v)$ because each message has its sender ID appended.) After which, the center of $C$ downcasts, simultaneously for all nodes $u \in C$, the corresponding computed aggregate packet to $u$.  

\paragraph{Analysis.} The above completes the description of the simulation. Next, we upper bound the runtime, maximum edge congestion and message complexity of the simulation (i.e., of $\mathcal{A}'$) in terms of the runtime (upper bound) $\TA$ of $\mathcal{A}$. We start by deriving the complexity bounds of preprocessing.

\begin{lemma}
\label{lem:tradeoffSimulation1Preprocessing}
    The preprocessing takes $\tilde{O}(n^{2-\eps})$ rounds and $\tilde{O}(m)$ messages. Moreover, it induces a maximum edge congestion of $\tilde{O}(1)$ for any non cluster edge.
\end{lemma}

\begin{proof}
    It is straightforward to see that the first step takes $\tilde{O}(D)$ rounds, sends $\tilde{O}(m)$ messages and that its maximum edge congestion is $\tilde{O}(1)$. 
    For the second step, 
    for any level of the cluster hierarchy, the aggregation induces no edge congestion on any non cluster edges and a maximum congestion of $\tilde{O}(n^{2-\eps})$ for cluster edges. Indeed, each cluster tree in the cluster hierarchy has no proper subtree with $\Omega(n^{1-\eps})$ nodes (see Corollary \ref{thm:PrunedBaswanaSenRunningProperties}), and each such node upcasts at most $O(n)$ messages. Moreover, since each cluster tree has depth $O(\kappa)$, the aggregation can be pipelined and the runtime (when considering that level of the cluster hierarchy only) is $O(n^{2-\eps} + \kappa) = O(n^{2-\eps})$. As for messages, each $\log n$ bits upcasted by some node generates at most $\kappa$ messages, since each cluster tree has depth $O(\kappa)$. Hence, the message complexity (when considering a single level of the cluster hierarchy) is $O(m \kappa)$. Over all levels of the cluster hierarchy, this amounts to another $O(\kappa)$ blow-up. Finally, the statement follows from the fact that $\kappa = O(\log n)$.
\end{proof}

Next, we upper bound the maximum edge congestion over all edges, and in particular over cluster edges, over any one phase. This is crucial to argue the correctness of the simulation.

\begin{lemma}
\label{lem:simulationPhaseCongestion}
    Consider any one phase of $\mathcal{A}'$. Then, the maximum congestion is $\tilde{O}(n)$ over cluster edges and $\tilde{O}(1)$ over other edges.
\end{lemma}

\begin{proof}
    During the simulation, communication over the cluster edges is done via the upcast and downcast operations. First, note that for every cluster $C$ and every upcast operation (resp., the downcast operation in the receive step), at most $\tilde{O}(1)$ messages are upcast to (resp., downcast from) the center per node (of $C$) over the cluster tree edges (and only them), where the $\tilde{O}(1)$ overhead comes from the size of the aggregate packets. Similarly, for every cluster $C$ and the downcast operation in the direct send sub-step, at most $\tilde{O}(1)$ messages are downcast from the center per node neighboring $C$, and over the cluster edges (and only them). Hence, at most $\tilde{O}(n)$ different messages transit through any given cluster edge per upcast operation (resp., downcast operation). Now, it suffices to bound the number of downcast and upcast operations a given cluster edge can participate. As any level of the pruned Baswana-Sen cluster hierarchy consists of disjoint clusters, any cluster edge can be in at most $O(\kappa) = \tilde{O}(1)$ different trees in the hierarchy, and thus can participate in at most $\tilde{O}(1)$ upcast and downcast operations. Thus, the maximum congestion over cluster edges is $\tilde{O}(n)$ per phase.

    As for the other edges, messages only transit over inter-communication edges during the simulation. For any level of the cluster hierarchy, from the description of the direct send sub-step, it is clear a single message transits over any given inter-communication edge. Hence, over the $O(\kappa) = \tilde{O}(1)$ levels of the cluster hierarchy, the maximum congestion over inter-communication edges is at most $\tilde{O}(1)$ per phase.
\end{proof}

Given that at most $\tilde{O}(n)$ messages transit through any one edge during any one phase of the simulation, and cluster trees have depth $\tilde{O}(1)$, we can pipeline these $\tilde{O}(n)$ messages messages (with some priority that depends on the level of a cluster the message is being upcast or downcast over) and obtain that $\tilde{O}(n)$ rounds suffice for all operations (in particular, the downcast and upcast operations) described above to complete successfully. By setting the number of rounds per phase accordingly, we get the following corollary. 

\begin{corollary}
\label{cor:simulationSuccessfulOperations}
    Consider any one phase of $\mathcal{A}'$. Then, the downcast and upcast operations within the phase terminate successfully within the phase's $\tilde{O}(n)$ rounds.
\end{corollary}

After which, we can argue that $\mathcal{A}'$ correctly simulates $\mathcal{A}$, even when it simulates the communication of $\mathcal{A}$ in an aggregate sense only.

\begin{lemma}
\label{lem:tradeoffSimulation1Correctness}
For any phase $p \in [1,\TA]$ and node $v \in V$, $v$ ends phase $p$ of $\mathcal{A}'$ in the same state it would end round $p$ of $\mathcal{A}$. 
\end{lemma}

\begin{proof}
We prove, by induction on $p$, a slightly different statement: for any phase $p \in [1,\TA]$ and node $v \in V$, $v$ starts phase $p$ of $\mathcal{A}'$ in the same state it would start round $p$ of $\mathcal{A}$. 
The base case is trivially guaranteed by the simulation's description. (Note that the base case is trivial, unlike the simulation in Section \ref{sec:simulation}; in the latter, nodes are simulated by cluster centers, and these centers must gather the input of all their simulated nodes for the base case to hold.) Hence, we now assume that the induction hypothesis holds for some phase $p \in [1,\TA-1]$, and we will prove the induction step by showing the following claim: for any node $u \in V$, if the aggregate of messages sent by nodes in $B_p(u)$ is $\mu$, then $u$ computes $\mu$ by the end of the compute step of phase $p$. Indeed, this suffices for each node $u \in V$ to simulate its local state change from round $p$ to round $p+1$ in the simulated algorithm $\mathcal{A}$.

We separate nodes in $B_p(u)$ into two node subsets $I_p(u)$ and $D_p(u)$, where the first contains all nodes in $B_p(u)$ that are incident to some inter-communication edge leading to some cluster containing $u$, as well as all nodes in $B_p(u)$ that share a cluster with $u$, and the second contains all other nodes in $B_p(u)$. By property (c) holding for pruned Baswana-Sen cluster hierarchies (see Theorem \ref{theorem:BaswanaSenProperties} and Corollary \ref{thm:PrunedBaswanaSenRunningProperties}), it holds that for any node $v \in D_p(u)$, there must be an inter-communication edge incident to $u$ leading to a cluster containing $v$, but not to $v$ itself. 

First, the indirect send sub-step ensures that for any node $v \in I_p(u)$, the message $m_v$ sent by $v$ is either sent over to some cluster $C$ containing $u$, or $v$ belongs to a cluster containing $u$. In which case, as long as the upcast and downcast operations within the receive step completes successfully---and this holds true by Corollary \ref{cor:simulationSuccessfulOperations}---the center of $C$ computes an aggregate packet with $m_v$ as one of the inputs, and sends it to $u$. In summary, for any node $v \in I_p(u)$, $u$ receives some (not necessarily unique) aggregate packet, whose inputs to $agg_{u,p}$ contains $m_v$ and possibly other messages, but only from nodes in $B_p(u)$.

Next, the direct send sub-step ensures that for any node $v \in D_p(u)$, the message $m_v$ is upcasted to the center of all clusters containing $v$. Then, there exists at least one center $C$ that computes an aggregate packet, using $agg_{u,p}$ over inputs that are exactly the messages from nodes in $B_p(u) \cap C$, and downcasts that aggregate packet to some node $w$. Once the downcast operation completes successfully---and this holds true by Corollary \ref{cor:simulationSuccessfulOperations}---then $w$ in turn sends it over an incident inter-communication edge (directly) to $u$. In summary, for any node $v \in D_p(u)$, $u$ receives some aggregate packet, whose inputs to $agg_{u,p}$ contains $m_v$, and possibly other messages, but only from nodes in $B_p(u)$.

Finally, in the Compute step, node $u$ computes the aggregate of all packets received during direct send and receive. Since doing so aggregates all messages in $I_p(u)$ and $U_p(u)$, and $B_p(u) = I_p(u) \cup U_p(u)$, we get that $u$ computes $\mu$. By the definition of an aggregate-based algorithm (see Definition \ref{def:aggregationBasedAlgorithm}), each node $u \in V$ can simulate its local state change from round $p$ to round $p+1$ in the simulated algorithm $\mathcal{A}$.

As a final step, to obtain the lemma statement, we take the above statement for $p = \TA$ and apply the above claim one more time.
\end{proof}

Now, we can show runtime, message complexity and maximum edge congestion upper bounds for the simulation part.

\begin{lemma}
\label{lem:tradeoffSimulation1Simulation}
    The simulation part takes $\tilde{O}(\TA \cdot n)$ rounds, and $\tilde{O}(\TA \cdot n^{1+\eps})$ messages with high probability. Moreover, it induces a maximum edge congestion of $\tilde{O}(\TA)$ for any non cluster edge.
\end{lemma}
\begin{proof}
    The runtime follows directly from the simulation's description, and the maximum congestion per edge statement from adding up over all $\TA$ phases the congestion given by Lemma \ref{lem:simulationPhaseCongestion}. 
    Finally, we bound the message complexity of any given phase $p \in [1, \TA]$ by $\tilde{O}(n^{1+\eps})$ in three parts. First, the number of messages sent during the indirect send sub-step is $\tilde{O}(n^{1+\eps})$ with high probability since there are at most $\tilde{O}(n^{1+\eps})$ inter-communication edges in the pruned Baswana-Sen cluster hierarchy with high probability (see property (b) in Theorem \ref{theorem:BaswanaSenProperties} and Corollary \ref{thm:PrunedBaswanaSenRunningProperties}) and at most $O(1)$ (aggregate) packets (thus $\tilde{O}(1)$ messages) transit over any one such edge. Second, the number of messages sent during the direct send sub-step is also $\tilde{O}(n^{1+\eps})$. Indeed, the upcast and downcast operations use $\tilde{O}(n)$ messages, over all clusters in all levels of the hierarchy, since each node upcasts a single message, and at most one (aggregate) packet (thus $\tilde{O}(1)$ messages)  is downcast per node neighboring the downcasting cluster. After which, these aggregate packets are sent over inter-communication edges, resulting in at most $O(n^{1+\eps})$ messages with high probability. Third and finally, the number of messages sent during the receive is $\tilde{O}(n)$ since each cluster (in all levels of the hierarchy) executes one upcast and one downcast operation, where at most one aggregate packet (thus $\tilde{O}(1)$ messages) is upcast from, and downcast to, any given node in the cluster. When adding up the message complexity of the three parts above, over all phases, we get $\tilde{O}(\TA \cdot n^{1+\eps})$ message complexity with high probability.
\end{proof}

Finally, we can combine the previous statements (Lemma \ref{lem:tradeoffSimulation1Correctness}, and Lemmas \ref{lem:tradeoffSimulation1Preprocessing} and \ref{lem:tradeoffSimulation1Simulation}), to prove our simulation has the desired correctness, runtime, message complexity and maximum edge congestion properties---see Theorem \ref{thm:tradeoffSimulation1} below. 

\simulationGeneral*

\subsubsection{Improved Simulation for $\eps \in [1/2,1]$:} This simulation takes as input a pruned Baswana-Sen cluster hierarchy with parameter $\eps \geq 1/2$ (i.e., with at most 3 levels). Note that the cluster hierarchy comprises a partition $\mathcal{C}_0$ of $V$ into singleton clusters, and possibly a partition $\mathcal{C}_1$ of some subset $V_1 = V \setminus L_1$ into depth 1 clusters---which we call \textit{star clusters}. Moreover, the following properties (see Lemma \ref{lem:starClusterHierarchyProperties}) can be easily shown to hold with high probability.

\begin{lemma}
\label{lem:starClusterHierarchyProperties}
    Any pruned Baswana-Sen cluster hierarchy with parameter $\eps \geq 1/2$ satisfies the following properties with high probability:
    \begin{itemize}
        \item For any node $v \in L_1$, $v$ has at most $O(n^{\eps} \log n)$ incident edges in $G$, and all are in $F_1$.
        \item $|\mathcal{C}_1| =  \tilde{O}(n^{1-\eps})$.
    \end{itemize}
\end{lemma}

The improved simulation also consists of a preprocessing part followed by a simulation part. The preprocessing part is identical to the preprocessing for general $\eps$ (see Sec.~\ref{sec:general}), whereas the simulation part contains some key differences. Recall that the simulation part for general $\eps$ works in phases, each simulating one round of $\mathcal{A}$ through the three \textit{send, receive and compute} steps. Here, the key differences are that each phase takes significantly less rounds---which is allowed since we limit the congestion over cluster edges to $\tilde{O}(n^{1-\eps})$ per phase---and the send step is implemented differently (but the receive and compute steps remain identical).

More concretely, this send step is implemented as follows. Consider some node $v$ that broadcasts a message $m_v$ in round $p$ of $\mathcal{A}$. 
\begin{itemize}
    \item If $v \in L_1$, then $v$ simply sends $m_v$ over all incident inter-communication edges---this implements the direct send sub-step (and there is no indirect send sub-step for $L_1$ nodes).
    \item Otherwise, $v$ is part of a star cluster, say $C$, in which case $v$ sends $m_v$ to its parent (i.e., the center of $C$). After receiving messages from all broadcasting nodes in $C$ (i.e., from $B_p \cap C$), the center executes the following local computations for any neighboring cluster $C' \in \mathcal{C}_1$: 
        \begin{itemize}
            \item The center computes the set $R_p(C,C')$ comprising the neighbors in $C'$ of nodes in $B_p \cap C$---i.e., nodes in $N(B_p \cap C) \cap C'$---using its knowledge of all edges incident to nodes in $C$. (Note that $C' \cap C = \emptyset$.)
            \item Next, the center computes a maximal matching $M(C,C') \subseteq E$ between nodes in $B_p \cap C$ and $R_p(C,C')$. 
            \item For each edge $e = (w,u) \in M(C,C')$, where w.l.o.g. $w \in C$, the center computes a message $m_1(e)$ containing the ID of $w$ as well as the message $m_w$ broadcasted by $w$, and a aggregate packet $m_2(e)$ which consists of $agg_{u,p}$ applied to all messages sent by nodes in $B_p(u) \cap C$---the resulting aggregate packet contains $\tilde{O}(1)$ bits, by Definition \ref{def:aggregationBasedAlgorithm}, and thus can be transmitted using $\tilde{O}(1)$ only. 
        \end{itemize}
        After these local computations, for any $C' \in \mathcal{C}_1$ and any edge $e \in M(C,C')$, the center sends $m_1(e)$ and $m_2(e)$ to the endpoint of $e$ in $C$. (Note that this can amount to at most $\tilde{O}(n^{1-\eps})$ messages on the edge from the center to $w$, w.h.p., since $|\mathcal{C}_1| = \tilde{O}(n^{1-\eps})$ w.h.p.) Upon reception, that endpoint node sends two different messages, containing respectively $m_1(e)$ and $m_2(e)$, through $e$. Over the entire cluster $C$, the first messages implement the indirect send sub-step, whereas the second messages (or packets) implement the direct send sub-step.
\end{itemize} 

\paragraph{Analysis.} We start by showing that a pruned Baswana-Sen cluster hierarchy with parameter $\eps \geq 1/2$ leads to significantly faster preprocessing.

\begin{lemma}
\label{lem:tradeoffSimulation2Preprocessing}
    The preprocessing part takes $\tilde{O}(n)$ rounds and $\tilde{O}(m)$ messages. Moreover, it induces a maximum edge congestion of $\tilde{O}(1)$ for any non cluster edge.
\end{lemma}

\begin{proof}
    The improvements to the preprocessing lie in the runtime of the second step. Indeed, all nodes may belong to at most one non-singleton cluster, which is of depth 1, and may send at most $\tilde{O}(n)$ to their parent---the center---in this cluster. Hence, the congestion over cluster edges is only $\tilde{O}(n)$, and the runtime directly follows.
\end{proof}

Next, we show that the above described (modified) send step, when combined with the properties of the pruned Baswana-Sen cluster hierarchies for $\eps \geq 1/2$, leads to improved maximum edge congestion (over cluster tree edges) per simulation phase.

\begin{lemma}
\label{lem:simulationPhaseCongestion2}
    Consider any phase of $\mathcal{A}'$. The maximum congestion is $\tilde{O}(n^{1-\eps})$ over cluster trees edges (w.h.p.), and $\tilde{O}(1)$ over other edges.
\end{lemma}

\begin{proof}
    Here, it suffices to focus on the maximum congestion over cluster tree edges. Note that we only need to consider nodes belonging to star clusters, i.e., nodes in $V_1$. 
    First, we bound that congestion for the send step. Note that every node sends at most one message to the center of the star cluster. After which, the center sends at most $\tilde{O}(n^{1-\eps})$ messages to any node in $C$, because $|\mathcal{C}_1| = \tilde{O}(n^{1-\eps})$ with high probability (by Lemma \ref{lem:starClusterHierarchyProperties}), and for any $C' \in \mathcal{C}_1$, then any two edges in the matching $M(C, C')$ (on the bipartite graph induced by $B_p \cap C$, and $R_p(C,C')$) must have different endpoints in $C$. In summary, the (cluster tree) edge congestion over the send step is $\tilde{O}(n^{1-\eps})$ with high probability. 
    
    Next, we bound the cluster tree edge congestion over the receive step. To do so, we first point out that, for any star cluster $C$, at most one message containing some node's ID (and thus corresponding to an indirect send) is sent to any node $u \not\in C$. Indeed, for any cluster $C' \in \mathcal{C}_1$, once again any two edges of the matching $M(C, C')$ (on the bipartite graph induced by $B_p \cap C$, and $R_p(C, C')$) must lead to different nodes in $C'$. Moreover, the clusters of $\mathcal{C}_1$ are disjoint (see Corollary \ref{thm:PrunedBaswanaSenRunningProperties}), hence each node receives at most one node's ID per other cluster. Since we also know that with high probability, there are at most $\tilde{O}(n^{1-\eps})$ star clusters in the cluster hierarchy (by Lemma \ref{lem:starClusterHierarchyProperties}), it follows that each node $u$ may receive up to at most $\tilde{O}(n^{1-\eps})$ IDs (from an indirect send sub-step), and thus $u$ sends that many messages to its parent, as well as an additional message with its ID if $u$ broadcasts in that phase, in the receive step. In summary, the (cluster tree) edge congestion over the receive step is $\tilde{O}(n^{1-\eps})$ with high probability. Finally, it is straightforward to show that the edge congestion is $\tilde{O}(1)$ over all other edges, and thus the lemma statement follows.
\end{proof} 

Given that at most $\tilde{O}(n^{1-\eps})$ messages transit through any one edge during any one phase of the simulation with high probability, it is straightforward to see that $\tilde{O}(n^{1-\eps})$ rounds suffice for all operations (in particular, the downcast and upcast operations) of each simulation phase to complete successfully, with high probability. Setting the number of rounds per phase accordingly, we get the following corollary.

\begin{corollary}
\label{cor:simulationSuccessfulOperations2}
    Consider any one phase of $\mathcal{A}'$. With high probability, all operations within the phase terminate successfully within the allocated $\tilde{O}(n^{1-\eps})$ rounds. 
\end{corollary}

We are now ready to argue that $\mathcal{A}'$ correctly simulates $\mathcal{A}$. Note that the communication of $\mathcal{A}$ is simulated in an aggregate sense only.

\begin{lemma}
\label{lem:tradeoffSimulation2Correctness}
With high probability, for any phase $p \in [1,\TA]$ and node $v \in V$, $v$ ends phase $p$ of $\mathcal{A}'$ in the same state it would end round $p$ of $\mathcal{A}$. 
\end{lemma}

\begin{proof}
We prove, by induction on $p$, a slightly different statement: with high probability, it holds that, for any phase $p \in [1,\TA]$ and node $v \in V$, $v$ starts phase $p$ of $\mathcal{A}'$ in the same state it would start round $p$ of $\mathcal{A}$.

Since this follows along the lines of the proof of Lemma \ref{lem:tradeoffSimulation1Correctness}. We consider only the induction step from some phase $p \in [1,\TA-1]$ to phase $p+1$, which we show by proving the following claim: for any node $u \in V$, if the aggregate of messages sent by nodes in $B_p(u)$ is $\mu$, then $u$ computes $\mu$ by the end of the compute step of phase $p$. Indeed, this suffices by Definition \ref{def:aggregationBasedAlgorithm} for each node $u \in V$ to simulate its local state change from round $p$ to round $p+1$ in the simulated algorithm $\mathcal{A}$.

We separate nodes in $B_p(u)$ into four node subsets $L_1(u) = L_1 \cap B_p(u)$, $C_p(u)$, $M_p(u)$ and $U_p(u)$ where $C_p(u)$ are nodes in $B_p(u)$ that share a cluster with $u$, $M_p(u)$ denote nodes who were matched by their cluster center to some node $w \neq u$ in the star cluster of $u$, and $U_p(u)$ all other nodes in $B_p(u) \cap V_1$. For the first subset, it holds by Lemma \ref{lem:starClusterHierarchyProperties} that any node $v \in L_1(u)$ is connected to $u$ via an inter-communication edge, hence in the send step, $u$ receives the message of $v$. For the second and third subsets (which only concerns $u$ if $u$ is part of some star cluster), the messages from nodes in $M_p(u)$ are sent over to the cluster of $u$. After which, during the receive step, these messages, along with the messages from $C_p(u)$, are sent to the cluster center, aggregated using the center's knowledge of all incident edges to the cluster---that is, using $agg_{u,p}$ on the messages from $M_p(u) \cup C_p(u)$, and only these---and finally the center informs $u$ of the aggregated packet successfully with high probability (by Corollary \ref{cor:simulationSuccessfulOperations2}). In other words, node $u$ receives an aggregate packet where $agg_{u,p}$ was applied to exactly the messages of (nodes in) $C_p(u) \cup M_p(u)$ with high probability. Finally, for the fourth subset, it suffices to remark that for any cluster $C$ such that $C \cap U_p(u)$ is non-empty, the center of $C$ must have computed an edge from $C$ to $u$ in its matching $M(C, C')$ where $C' \in \mathcal{C}_1$ denotes the star cluster containing $u$. (Otherwise, an edge between some node in $U_p(u) \cap C$ and $u$ can be added to $M(C, C')$.) Then, the center of $C$ computes an aggregate packet containing $agg_{u,p}$ applied to all the messages sent by (nodes in) $B_p(u) \cap C \supseteq U_p(u) \cap C$, and only those, and then the packet is sent directly to $u$ over the matched edge.

Subsequently, in the Compute step, node $u$ computes the aggregate of all messages received during direct send---corresponding to $L_1(u)$ and $U_p(u)$---and receive---corresponding to $C_p(u)$ and $M_p(u)$ with high probability---thus $u$ computes $\mu$. As argued above, the claim suffices to obtain the induction statement, in particular for $p=\TA$. After which, it suffices to apply the claim one last time to get the lemma statement.
\end{proof}

Now, we can show runtime, message complexity and edge congestion upper bounds for the simulation part.

\begin{lemma}
\label{lem:tradeoffSimulation2Simulation}
    $\mathcal{A}'$ simulates $\mathcal{A}$ correctly (w.h.p.).
    This simulation takes $\tilde{O}(\TA \cdot n^{1-\eps})$ rounds, and (w.h.p.) sends $\tilde{O}(\TA \cdot n^{1+\eps})$ messages. Moreover, it induces a maximum edge congestion of $\tilde{O}(\TA)$ for any non cluster edge.
\end{lemma}

\begin{proof}
    The runtime follows directly from the simulation's description (with the reduced number of rounds per phase), and the maximum congestion statement from adding up over all $\TA$ phases the (reduced) congestion given by Lemma \ref{lem:simulationPhaseCongestion2}.
    
    It remains to bound the message complexity of any given phase $p \in [1, \TA]$ by $\tilde{O}(n^{1+\eps})$ with high probability. First, note that nodes in $L_1$ send in total at most $\tilde{O}(n^{1+\eps})$ messages (per phase) with high probability (by Lemma \ref{lem:starClusterHierarchyProperties}) during the send step, and none during the receive step. 
    
    Now, consider the nodes in $V \setminus L_1$, that are necessarily in star clusters. During the send step each such node sends at most one message to its cluster center, adding up to $O(n)$ messages. Next, for each cluster $C \in \mathcal{C}_1$, the cluster center computes one maximal matching per neighboring cluster in $\mathcal{C}_1$, and sends one message and one aggregate packet per edge over all these matchings. Since the maximal matchings are all between (some nodes in) the $C$ and a neighboring cluster $C' \in \mathcal{C}_1$, and all clusters of $\mathcal{C}_1$ are disjoint, this implies that the center of $C$ sends $\tilde{O}(n)$ messages. Since $|\mathcal{C}_1| = \tilde{O}(n^{1-\eps})$ by Lemma \ref{lem:starClusterHierarchyProperties}, at most $\tilde{O}(n^{2-\eps})$ are sent per cluster centers during the send step. Given that we assume $\eps \geq 1/2$, this implies centers send $\tilde{O}(n^{1+\eps})$ messages during the send step.    Finally, each such message and aggregate packet is forwarded at most once (over matched edges), adding up to $\tilde{O}(n^{1+\eps})$ messages over the send step. Then, during the receive step, only (some of) these messages are forwarded to the centers of the clusters, along with an additional message per upcasting node. The center in turn executes some aggregation operation and sends at most one aggregate packet per node in its cluster, thus the receive step also takes $\tilde{O}(n^{1+\eps})$ messages only. This suffices to upper bound the message complexity per phase.
\end{proof}

Finally, we can combine the previous statements (Lemma \ref{lem:tradeoffSimulation2Correctness}, and Lemmas \ref{lem:tradeoffSimulation2Preprocessing} and \ref{lem:tradeoffSimulation2Simulation}), to prove our simulation has the desired correctness, runtime, message complexity and maximum edge congestion properties---see Theorem \ref{thm:tradeoffSimulation2} below. 

\simulationFaster*

\subsection{Message-Time Trade-off for Unweighted APSP}

Consider the unweighted APSP problem, where the goal is for each node to output all of its distances to all other nodes. 
We give an unweighted APSP algorithm that achieves a natural 
message-time trade-off between the round-optimal (but non message-optimal) algorithm (see \cite{bernstein2021distributed}) and the message-optimal (but not runtime-optimal) algorithm (see Subsection \ref{subsec:simulationApplications1}). 

First, we point out that for $\eps = O(1 /\log n)$, the trade-off statement reduces to solving unweighted APSP $\tilde{O}(n^2)$ round and message complexities. This can be achieved by using the (weighted) APSP algorithm obtained in Section \ref{subsec:simulationApplications1}. 

The remainder of the trade-off requires more sophisticated techniques, including the simulations described in Section \ref{subsec:simulationTradeoff}. In what follows, we first show how to execute some polynomial in $n$ amount of BFS computations (possibly of limited depth, depending on the range of $\eps$) simultaneously, while getting a message-time trade-off. 

More precisely, a BFS tree computation (up to some depth $d$) should ensure that each node $v$ know its parent in some BFS of depth $d$ rooted in some node $u$, if the distance between $u,v$ is at most $d$. Note that typically, the random delays technique or congestion plus dilation framework---see Section \ref{section:congestionPlusDilation}---can be used to run multiple BFS computations time-efficiently, in a way that takes advantage of the low congestion per edge induced by a BFS computation. Here, we show how to simulate multiple BFS algorithms with significantly lower message complexity, but at the cost of higher runtime; here also, we leverage the low congestion per edge induced by a BFS computation. The more precise statement is captured by the following two auxiliary Lemmas~\ref{lem:multiBFS1} and \ref{lem:multiBFS2}, which are play a key role in our unweighted APSP algorithm that achieves the remainder of the trade-off. 

Our next two results, Lemmas~\ref{lem:multiBFS1} and \ref{lem:multiBFS2}, both make use of Theorem~\ref{thm:congestPlusDilationBFS}, which assumes access to shared randomness.
We can simulate this in our setting as follows:
We first elect a leader and construct a global BFS tree $T$ rooted at the leader. 
Then, the leader generates a string $S$ of $\Theta(n \log n)$ private random bits, which are sent to all nodes in a pipelined-manner over the edges of $T$.
Each node with ID $i$ will use the $\Theta(\log n)$-length segment of $S$ starting at position $i\cdot \Theta(\log n)$ as its source of randomness for choosing a random delay as required by Theorem~\ref{thm:congestPlusDilationBFS}.
(Note that we can ensure that all node IDs are in the range $[n]$, by using $T$ and the leader to rename the nodes accordingly.)
It is straightforward to verify that the time and message complexity of implementing shared randomness per use of Theorem~\ref{thm:congestPlusDilationBFS} are bounded by $\tilde O(n)$ rounds and $\tilde{O}(n^2)$ messages respectively. 

\begin{lemma}
\label{lem:multiBFS1}
    For any $\eps \in [1/2,1]$, there exists a \congest{} algorithm that computes $n$ BFS trees (w.h.p.), using $\tilde{O}(n^{2-\eps})$ rounds and sending (w.h.p.) $\tilde{O}(n^{2+\eps})$ messages. 
\end{lemma}

\begin{proof}
    Consider $n$ independent BFS construction \bcongest{} algorithms.
    Then, we can use the random delays technique (by Theorem \ref{thm:congestPlusDilationBFS}) to obtain a \bcongest{} aggregation-based algorithm $\mathcal{A}$ executing all $n$ BFS constructions using only $\tilde{O}(n)$ rounds. Since $\mathcal{A}$ is a \bcongest{} aggregation-based algorithm, we can plug it into the second simulation from Section \ref{subsec:simulationTradeoff} (see Theorem \ref{thm:tradeoffSimulation2}) to obtain a randomized \congest{} algorithm $\mathcal{A}'$ that computes all $n$ BFS trees correctly with high probability. As for its round and message complexities, these are respectively $\tilde{O}(n \cdot n^{1-\eps}) = \tilde{O}(n^{2-\eps}) $ and (w.h.p.) $\tilde{O}(n \cdot n^{1+\eps}) = \tilde{O}(n^{2+\eps})$. 
\end{proof}

\begin{lemma}
\label{lem:multiBFS2}
    For any $\eps \in [\frac{1}{\Theta(\log n)},1/2]$, there exists a \congest{} algorithm that computes $n$ BFS trees up to depth $\tilde{O}(n^{1-\eps})$  (w.h.p.), using $\tilde{O}(n^{2-\eps})$ rounds and sending (w.h.p.) $\tilde{O}(n^{2+\eps})$ messages. 
\end{lemma}

\begin{proof}
Consider $n$ BFS explorations up to depth $\tilde{O}(n^{1-\eps})$. Then, we can partition the $n$ BFS explorations into $b_\eps = \lceil n^{\eps} \rceil$ independent batches of $\tilde{\Theta}(n^{1-\eps})$ BFS explorations, each up to depth $\tilde{O}(n^{1-\eps})$. Note that the dilation of each such BFS exploration is $O(n^{1-\eps})$, because of the limited depth. Hence, we can apply the random delays technique to each batch (see Theorem \ref{thm:congestPlusDilationBFS}) and thus for each, obtain a fast \bcongest{} aggregation-based algorithm. If we denote the resulting algorithm for any batch $i \in [1, b_\eps]$ by $\mathcal{A}_i$, then each $\mathcal{A}_i$ correctly computes $\tilde{\Theta}(n^{1-\eps})$ BFS explorations, each up to depth $\tilde{O}(n^{1-\eps})$, using $\tilde{\Theta}(n^{1-\eps})$ rounds. Moreover, $\mathcal{A}_i$ a \bcongest{} aggregation-based algorithm. 

Note that as described prior to Lemma \ref{lem:multiBFS1} above, getting the shared randomness to apply Theorem \ref{thm:congestPlusDilationBFS} once takes $\tilde O(n)$ rounds and $\tilde{O}(n^2)$ messages. Applying Theorem \ref{thm:congestPlusDilationBFS} a total of $b_\eps$ times leads to an $\tilde O(n^{1+\eps})$ rounds and $\tilde{O}(n^{2+\eps})$ messages overhead for the simulation. As $\eps \leq 1/2$, the round complexity is in fact $\tilde O(n^{2-\eps})$.

Next, we compute $b_\eps$ independent pruned Baswana-Sen cluster hierarchies $\mathcal{H}_1,\ldots,\mathcal{H}_{b_\eps}$, and since $\mathcal{A}_i$ is a \bcongest{} aggregation-based algorithm, we can simulate algorithm $\mathcal{A}_i$ (see Theorem \ref{thm:tradeoffSimulation1}) using cluster hierarchy $\mathcal{H}_i$. If we denote the resulting algorithm by $\mathcal{A}_i'$, then each $\mathcal{A}_i'$ correctly (w.h.p.) computes $\tilde{\Theta}(n^{1-\eps})$ BFS explorations, each up to depth $\tilde{O}(n^{1-\eps})$, running in $\tilde{O}(n^{2-\eps})$ time, using $\tilde{O}(n^2)$ messages (w.h.p.) and such that the maximum congestion is $\tilde{O}(n^{1-\eps})$ on any non cluster edge (with respect to the $\mathcal{H}_i$).

At this point, we use the congestion smoothing lemma (see Lemma \ref{lem:congestionSmoothing}) to show that the maximum congestion over any edge (that is a cluster edge in any of the $\mathcal{H}_1,\ldots,\mathcal{H}_{b_\eps}$) over $\mathcal{A}_1',\ldots, \mathcal{A}_{b_\eps}'$ is $\tilde{O}(n^{2-\eps})$. As for the other edges, these are non cluster edges in all $\mathcal{A}_1',\ldots, \mathcal{A}_{b_\eps}'$  (with respect to $\mathcal{H}_1,\ldots,\mathcal{H}_{b_\eps}$ respectively) and thus have congestion $\tilde{O}(n^{1-\eps} \cdot n^{\eps}) = \tilde{O}(n)$. In summary, the maximum edge congestion over all edges and all algorithms $\mathcal{A}_1',\ldots, \mathcal{A}_{b_\eps}'$ is $\tilde{O}(n^{2-\eps})$, and the dilation of any $\mathcal{A}_i'$ is $\tilde{O}(n^{2-\eps})$. 

Finally, we apply Theorem \ref{thm:congestionPlusDilation} on algorithms $\mathcal{A}_1,\ldots,\mathcal{A}_{b_\eps}$ to obtain a \congest{} algorithm $\mathcal{A}_f$. It follows that $\mathcal{A}_f$ computes all $n$ BFS explorations up to depth $\tilde{O}(n^{1-\eps})$, using $\tilde{O}(n^{2-\eps})$ rounds and sending (w.h.p.) $\tilde{O}(n^{2+\eps})$ messages. 
\end{proof}

With the above two lemmas, we can show how to obtain the remaining portions of the trade-off for unweighted APSP.

\paragraph{Trade-off for $\eps \in [1/2,1]$}
To obtain the message-time trade-off for $\eps \in [1/2,1]$, we first apply Lemma \ref{lem:multiBFS1}, to obtain an algorithm running $n$ independent BFS computations, using $\tilde{O}(n^{2-\eps})$ rounds and sending (w.h.p.) $\tilde{O}(n^{2+\eps})$ messages. Once all BFS computations are done, each node knows its distance to all other nodes by taking its depth in each of the computed BFS trees.

\paragraph{Trade-off for $\eps \in [\frac{1}{\Theta(\log n)},1/2]$}
This part of the trade-off is more complex. We first compute all distances between nearby pairs of nodes---that is, at most $\tilde{O}(n^{1-\eps})$ hops apart---by applying Lemma \ref{lem:multiBFS2}. We obtain an algorithm that runs $n$ independent BFS computations up to depth $\tilde{O}(n^{1-\eps})$, using $\tilde{O}(n^{2-\eps})$ rounds and sending (w.h.p.) $\tilde{O}(n^{2+\eps})$ messages. Upon termination, each node $v$ stores its depth in each BFS tree as its distance to the BFS root (whose ID $v$ knows). Since we consider an undirected graph, this clearly allows each node to compute its distance to all other nodes within $\tilde{O}(n^{1-\eps})$ hops. Next, we deal with pairs of nodes that are further apart.  

As for computing the distances between any two nodes at least $\tilde{\Omega}(n^{1-\eps})$ hops apart, we first sample $\tilde{\Theta}(n^{\eps})$ nodes, called \textit{landmarks}, uniformly at random in $V$, and from each such landmark, compute a BFS tree rooted at that landmark (without any of the techniques used above). Once all $\tilde{O}(n^{\eps})$ BFS trees have been computed, each landmark (or root) gathers the information of all of the edges in the BFS tree (i.e., the IDs of any BFS edge's two endpoints) via an upcast operation. Then, each landmark broadcasts the information of the BFS tree (i.e, all ID pairs corresponding to all edges in the BFS tree) to all nodes. Finally, each node $v$ computes its distance to any node $u$ via any of these $\tilde{O}(n^{\eps})$ BFS trees, and updates its distance to $u$ if either it had no prior computed distance to $u$, or its distance to $u$ in memory is higher. 

It is straightforward enough to see that with high probability, for any two nodes $u,v$ that are some $\tilde{\Omega}(n^{1-\eps})$ hops apart, both $u$ and $v$ end up computing their distance between each other.

\APSPTradeoff*

\begin{proof}

The statement for $\eps \in [0,\frac{1}{\Theta(\log n)}]$ follows directly from Theorem \ref{thm:apsp}. And when $\eps \in [1/2,1]$ instead, the statement follows by Lemma \ref{lem:multiBFS1} and the fact that we consider an undirected graph for the unweighted APSP problem.

Finally, we consider $\eps \in [\frac{1}{\Theta(\log n)},1/2]$. First, we consider the distances between any two nodes that are at most $\tilde{O}(n^{1-\eps})$ hops apart. Then, both nodes correctly (w.h.p.) compute the distance by Lemma \ref{lem:multiBFS2} and the fact that we consider an undirected graph for the unweighted APSP problem. Next, we consider distances between nodes at least $\tilde{\Omega}(n^{1-\eps})$ hops apart. Then, a a simple probabilistic argument shows that any path of at least $\tilde{\Omega}(n^{1-\eps})$ hops contains a landmark node with high probability, as any node becomes a landmark independently and with probability $1/\tilde{\Theta}(n^{1-\eps})$. Hence, with high probability, for any two nodes $u,v$, their shortest path contains a landmark node. All of the edges of the BFS rooted at that landmark node are sent to both $v$ and $u$, and thus they correctly compute their distance to each other. The correctness (w.h.p.) for the algorithm when $\eps \in [\frac{1}{\Theta(\log n)},1/2]$ follows. 
As for its round and message complexities, computing the distances between any two nearby nodes takes $\tilde{O}(n^{2-\eps})$ rounds and sending (w.h.p.) $\tilde{O}(n^{2+\eps})$ messages by Lemma \ref{lem:multiBFS2}. As for the distance computations between any two far away nodes, it is straightforward to see that running the $\tilde{O}(n^{\eps})$ BFS computations, as well as the $\tilde{O}(n^{\eps})$ upcasts and broadcasts of $\tilde{O}(n)$ bits, takes $\tilde{O}(n^{1+\eps})$ rounds and $\tilde{O}(n^{2+\eps})$ messages. Since this runs only when $\eps < 1/2$, we can upper bound the runtime by $\tilde{O}(n^{2-\eps})$ rounds. The statement follows.
\end{proof}

\section{Conclusions and Future Work}

In this paper, we first focused on obtaining message-optimal distributed algorithms for several well-studied graph problems that did not have prior any such algorithms. While these problems are very different, our results were all obtained by applying a unified framework to existing round-optimal algorithms for these problems. However, as noted earlier,  our message-optimal algorithms are not round-optimal. 

We then showed that for unweighted APSP, we can construct a family of algorithms that smoothly exhibits a message-time trade-off across the entire spectrum --- from
message optimality (on one end) to time optimality (on the other end).
To the best of our knowledge, this is the first such tradeoff result known for a fundamental graph-optimization problem.

Several key open questions remain.
Does APSP admit singularly-optimal algorithms,  i.e., algorithms optimal with respect to {\it both} message and time
complexities?  Another important question is whether one can use our framework to obtain message-optimal
algorithms for other problems, such as computing a maximum matching in general graphs and obtaining message-time tradeoffs for \textit{weighted} APSP, maximum matching, minimum cut etc.

\bibliographystyle{plain}
\bibliography{references,references1}

\appendix
\section{Applications of the Simulation Theorem}
\label{sec:apps}

We now discuss two applications of the simulation theorem from Section \ref{sec:simulation}. We obtain \congest{} algorithms with $\tilde O(n^2)$ message complexity for computing (1) exact bipartite maximum matching in Section \ref{sec:bipartiteMaxM}, and (2) neighborhood covers in Section \ref{sec:neighborhoodCovers}. 

\subsection{Bipartite Maximum Matching}
\label{sec:bipartiteMaxM}

We start by briefly reviewing classical graph-theoretic definitions. Given a matching $M$ in a graph $G = (V, E)$, an  \textit{$M$-alternating path} is a path in which the edges are alternately in $M$ and in $E \setminus M$. A vertex is called \textit{$M$-free} if there are no edges in $M$ incident on it. An $M$-alternating path is called \textit{$M$-augmenting} if its end points are both $M$-free.
Most algorithms for computing a maximum matching in a graph are based on Berge's Theorem \cite{BergePNAS57}, which says that a matching $M$ in a graph $G = (V, E)$ is maximum if and only if there exists no $M$-augmenting path in $G$. 
If a matching $M$ is not a maximum matching, then according to Berge's theorem, $G$ contains an $M$-augmenting path $P$. If such a path $P$ could be found, we can then construct a matching that is larger than $M$ (by one edge), by replacing $M$ by the symmetric difference $M \oplus P$. One can then repeat this process until we have constructed a matching $M$ such that $G$ contains no $M$-augmenting paths, which certifies that $M$ is a maximum matching.

Ahmadi, Kuhn, and Oshman \cite{AhmadiKuhnOshmanDISC2018} describe a \congest{} model algorithm for computing a maximum matching in bipartite graphs that runs in $O(n \log n)$ rounds. While this algorithm is described in the \congest{} model, we argue below informally that this algorithm works, as is, in the \bcongest{} model.

The algorithm starts with a preprocessing step in which it computes an upper bound $s$ on the size $s^*$ of a maximum matching. This can be done by first computing a maximal matching $\hat{M}$, computing its size, and using $s := 2|\hat{M}|$ as the upper bound on $s^*$. A maximal matching can be computed in $O(\log n)$ rounds in the \bcongest{} model by simply using the well-known randomized algorithm of Israeli and Itai \cite{IsraeliItaiIPL1986}.
The size of this matching can be computed and then broadcast  by electing a leader and computing a BFS tree rooted at the leader. All of this can be done in $O(Diam)$ rounds in the \bcongest{} model.
Once the upper bound $s$ is computed, the algorithm initializes a matching $M := \emptyset$ 
and proceeds in \textit{Phases} $i = 1, 2, 3, \ldots, n-1$, attempting to find a 
``short'' $M$-augmenting path in each phase. More specifically, in Phase $i$, the algorithm spends $O(s/(s - i))$ rounds searching for an augmenting path of length $O(s/(s - i))$. 
If such an augmenting path $P$ is found, we update $M$ to $M \oplus P$ to increase its size by 1. The motivation for searching for an augmenting path of this length comes from a corollary in \cite{AhmadiKuhnOshmanDISC2018}, stating that if $|M| = i$, then there is an $M$-augmenting path of length less than $2 \lceil s^*/(s^* - i)\rceil$. This corollary in turn follows from an observation about the existence of short augmenting paths that drives the well-known Hopcroft-Karp algorithm \cite{HopcroftKarp73} for maximum matching in the sequential setting.

We now describe a typical phase of the algorithm.
Every free node $f$ initiates an \textit{alternating path broadcast} by sending $\text{ID}_f$ to all neighbors. These free nodes will not send any more broadcast messages during this phase of the algorithm. All other nodes propagate the first broadcast they receive along alternating paths; in odd rounds the broadcast is propagated along edges that are not in the matching, and in even rounds it is propagated along matching edges. If nodes receive multiple broadcasts at the same tie, they break ties using the IDs of the source of the broadcast. 
Two neighboring nodes $u$ and $v$ find an augmenting path in a round $r$ if in this round, these two nodes propagate their broadcast to each other along the edge $(u, v)$. Note that if this event occurs, then it must be the case that the sources of the broadcasts that have reached $u$ and $v$ are distinct. This implies that an $M$-augmenting path has been detected by nodes $u$ and $v$. In the rest of the phase, this information (about the detection of an $M$-augmenting path) is propagated back up the tree until it reaches the two $M$-free nodes that are sources of the broadcasts. The authors of \cite{AhmadiKuhnOshmanDISC2018} show a simple way of labeling each $M$-augmenting path that is discovered and if a node is informed about multiple $M$-augmenting paths during backward propagation, it only sends up information about the path with lexicographically smaller label. This ensures that even during the backward propagation, each node needs to send at most one broadcast. Finally, 
each $M$-free node will zero or more $M$-augmenting path labels for
which it is an endpoint. Each $M$-free node picks an incident $M$-augmenting path with the smallest label and does forward propagation along this path to confirm with the other end point that this path is available for augmenting.

\bipartiteMaxM*
\begin{proof}
The algorithm described above has round complexity $\tilde{O}(n)$ and broadcast complexity $O(n^2)$. 
The latter simply follows from the fact that the broadcast complexity of each phase is $O(n)$ (in fact $O(1)$ broadcasts per node per phase) and there are $n$ phases. 
For the maximum matching problem, the input size $In = O(m + n)$ and the output size is $O(n)$. Thus, by using Theorem \ref{thm:messageEfficientSimulation}, we see that this algorithm can be simulated in $\tilde{O}(n^2)$ rounds, with message complexity $\tilde{O}(n^2)$, in the \congest{} model. 
\end{proof}

\subsection{Neighborhood Covers}
\label{sec:neighborhoodCovers}

For positive integers $k$ and $W$, a collection of trees $\mathcal{C}$ is called a \emph{$(k,W)$-sparse neighborhood cover}, if (1) every tree in $\mathcal{C}$ has depth at most $O(W \cdot k)$, (2) each vertex $v$ appears in at most $\tilde O(k \cdot n^{1/k})$ different trees, and (3) there exists some tree in $\mathcal{C}$ that contains the entire $W$-neighborhood of $v$.
Neighborhood covers were first introduced in~\cite{awerbuch1993near} and the state-of-the-art algorithm for the $\congest$ model was given by Elkin in \cite{elkin2006faster}, which takes $O(k^2 \cdot n^{1/k} \cdot W \cdot \log n)$ rounds, succeeds with high probability, and has a message complexity of $O(m \cdot k \cdot n^{1/k}\cdot \log n)$. Combined with our message-efficient simulation, we obtain the following result: 

\neighborhoodCover*
\begin{proof}
The algorithm given by Elkin in \cite{elkin2006faster} proceeds in a sequence of phases, where each phase consists of a randomly chosen subset of the currently ``uncovered'' nodes initiating a BFS-exploration of depth at most $O(k \cdot W)$. These BFS-explorations can be implemented using broadcast communication. Moreover, Lemma A.5 of \cite{elkin2006faster} shows that each vertex is visited by at most $O(n^{1/k}\log n)$ BFS-explorations in a single phase with high probability. As a result, each phase has a broadcast complexity of $O(k \cdot n^{1+1/k} \cdot W\cdot \log n)$.
Since there are $k$ phases overall and the algorithm's inputs and outputs are of size $O(m) = O(n^2)$, Corollary~\ref{cor:messageEfficientSimulation} implies that the algorithm can be simulated in $\tilde{O}(n^2)$ rounds with $\tilde{O}(n^2 + k^2 n^{1+1/k} W) = \tilde{O}(n^2)$ messages in the \congest{} model.
\end{proof}

\end{document}